\documentclass[11pt]{article}
\pdfoutput=1
\usepackage{graphicx}
\usepackage[english]{babel}
\usepackage{amsmath,amssymb,amsthm}
\usepackage[outdir=./]{epstopdf}
\usepackage{tikz}
\usepackage{pifont}
\usepackage{url}
\usepackage{romanbar}
\usepackage{subfig}
\usepackage{listings}
\usepackage[labelfont=bf]{caption}

\allowdisplaybreaks

\newcommand{\R}{\operatorname{right}}
\renewcommand{\L}{\operatorname{left}}
\renewcommand{\O}{\operatorname{outer}}

\newtheorem{lemma}{Lemma}[section]

\newtheorem{proposition}[lemma]{Proposition}

\newtheorem{remark}[lemma]{Remark}
\usetikzlibrary{patterns,arrows,decorations.markings}

\numberwithin{equation}{section}
\newcommand{\Ri}{\Romanbar{1}}
\newcommand{\Ro}{\Romanbar{2}}
\newcommand{\Rright}{\Romanbar{3}}
\newcommand{\Rleft}{\Romanbar{4}}

\title{Construction and implementation of asymptotic expansions for Jacobi--type orthogonal polynomials}
\author{
A. Dea\~{n}o\footnote{alfredo.deanho@uc3m.es}\\
Department of Mathematics\\
Universidad Carlos III de Madrid, Spain
\and
D. Huybrechs\footnote{daan.huybrechs@cs.kuleuven.be}\\
Department of Computer Science\\
KU Leuven, Belgium
\and
Peter Opsomer\footnote{peter.opsomer@cs.kuleuven.be (corresponding author)}\\
Department of Computer Science\\
KU Leuven, Belgium
}

\begin{document}

\maketitle
\begin{abstract}
We are interested in the asymptotic behavior of orthogonal polynomials of the generalized Jacobi type as their degree $n$ goes to $\infty$. These are defined on the interval $[-1,1]$ with weight function 
\begin{equation}
	w(x)=(1-x)^{\alpha}(1+x)^{\beta}h(x), \quad \alpha,\beta>-1 \nonumber
\end{equation}
and $h(x)$ a real, analytic and strictly positive function on $[-1,1]$. This information is available in the work of Kuijlaars, McLaughlin, Van Assche and Vanlessen \cite{KMcLVAV}, where the authors use the Riemann--Hilbert formulation and the Deift--Zhou non-linear steepest descent method. 

We show that computing higher-order terms can be simplified, leading to their efficient construction. The resulting asymptotic expansions in every region of the complex plane are implemented both symbolically and numerically, and the code is made publicly available.

The main advantage of these expansions is that they lead to increasing accuracy for increasing degree of the polynomials, at a computational cost that is actually independent of the degree. In contrast, the typical use of the recurrence relation for orthogonal polynomials in computations leads to a cost that is at least linear in the degree. Furthermore, the expansions may be used to compute Gaussian quadrature rules in $\mathcal{O}(n)$ operations, rather than $\mathcal{O}(n^2)$ based on the recurrence relation.

\end{abstract}

\section{Introduction}\label{Sintro}

In this paper we are interested in the symbolic implementation and numerical computation of asymptotic expansions for monic polynomials $\pi_n(x)$ that are orthogonal with respect to a Jacobi-type weight function on the interval $[-1,1]$:
\begin{equation}
	\int_{-1}^1 \pi_n(x)\pi_k(x)w(x)dx=0, \qquad k=0,1,\ldots,n-1, \nonumber
\end{equation}
with
\begin{equation}
	w(x)=(1-x)^{\alpha}(1+x)^{\beta}h(x), \qquad \alpha,\beta>-1, \label{wx}
\end{equation}
and where $h(x)$ is a real analytic and strictly positive function on $[-1,1]$. 

The work of Kuijlaars, McLaughlin, Van Assche and Vanlessen \cite{KMcLVAV} provides complete asymptotic information of the large $n$ behavior of $\pi_n(x)$, together with associated quantities such as recurrence coefficients $\alpha_n$ and $\beta_n$ of the three term recurrence relation
\begin{equation}\label{TTRR}
\pi_{n+1}(x) = (x-\alpha_n)\pi_{n}(x) - \beta_n\pi_{n-1}(x),
\end{equation}
as well as leading term coefficients and Hankel determinants. These results are obtained using the Riemann--Hilbert formulation for $\pi_n(x)$, see the seminal paper of Fokas, Its and Kitaev, \cite{FIK}, and the steepest descent method due to Deift and Zhou, \cite{Deift,DZ}. This procedure gives three types of asymptotic expansions: outer asymptotics, valid for $x\in\mathbb{C}\setminus[-1,1]$, inner asymptotics, for $x \in (-1,1)$ but away from the endpoints, and boundary asymptotics, for $|x \mp 1|<\delta$, for some fixed $\delta>0$.

The purpose of this paper is to provide an automatic and efficient implementation (symbolic and numerical) of the asymptotic expansions for these Jacobi--type polynomials $\pi_n(x)$ presented in \cite{KMcLVAV}. In this reference, the leading order terms are given explicitly, and we detail the derivation of higher-order terms. It only requires elementary numerical techniques: in particular, it does not need any evaluation of special functions. We will also deal with some specific and non--trivial issues that arise when implementing the formulation of \cite{KMcLVAV} in the complex plane in a symbolic and numerical setting.  All formulas are implemented in {\sc Maple}, {\sc Julia} and {\sc Matlab} and are available on the software webpage of our research group: \url{http://nines.cs.kuleuven.be/software/JACOBI}.

From a computational point of view, the asymptotic expansions given in \cite{KMcLVAV} present two important advantages: they become increasingly accurate as the degree $n$ becomes large, and the computing time is essentially independent of $n$. In this sense, they compare favourably to other methods to compute $\pi_n(x)$, such as the use of the recurrence relation \eqref{TTRR}. For this reason, the idea of using asymptotic expansions for the computation of orthogonal polynomials has already been present in the literature for some time, either using approximations coming from integral representations or differential equations (see the work of Hale and Townsend \cite{HT,HT2} and references therein), or more recently from Riemann--Hilbert problems (see for instance the work of Olver and Trogdon \cite{OT_RH,OT_RH2,OT_RH3,olver,trog}. Additionally, these expansions can also be used to construct Gaussian quadrature rules with a high number of points: we refer the reader to \cite{glaser,bogaert,HT,peter,BogaertIterationFree} and \S\ref{Sremarks}.

Observe that when $h(x)\equiv 1$ we have the standard Jacobi polynomials, whose strong asymptotic behavior is well known, see for instance the classical monograph by Szeg\H{o}, \cite[Chapter VIII]{Szego} or the more recent one of Ismail \cite[Chapter 4]{Ismail}. Other extensions of this framework studied in the literature include a weight with an algebraic singularity and a discontinuity inside the interval of orthogonality:
\begin{equation}
	w(x)=(1-x)^{\alpha}(1+x)^{\beta}h(x)|x_0-x|^\gamma \Xi_c(x), \nonumber
\end{equation}
where $\alpha,\beta,\gamma >-1$, $x_0 \in (-1,1)$ and $\Xi_c(x)$ is a step function, see \cite{FMFS_asymp,FMFS_Magnus}. In this case a separate asymptotic analysis is needed in a neighborhood of $x=x_0$. It is also possible to consider Jacobi polynomials with non--standard parameters $\alpha$ and $\beta$, see for instance \cite{KMF_Jacobi,KMFO_Jacobi}, but in this article we will restrict ourselves to the classical case, with weight function \eqref{wx}.

\section{Asymptotic expansions for Jacobi--type polynomials}\label{Sasy}

In this section we present large $n$ asymptotic expansions for the orthogonal polynomials and related quantities. We distinguish between several regions in the complex plane, shown in Figure \ref{Fregions}, 
since the polynomial $\pi_n$ exhibits different asymptotic behaviour in $\mathbb{C}$:
\begin{itemize}
 \item a complex neighbourhood of the interval $(-1,1)$ excluding the endpoints, subsequently called the `lens'  (region \Ri)
 \item two disks around the endpoints $\pm 1$, called the right and left disk (regions \Rright~and \Rleft)
 \item  and the remainder of the complex plane, the `outer region' (region \Ro).
\end{itemize}
\begin{remark} \label{Rheuristics}
Mathematically, the regions are of arbitrary size and, depending on how they are chosen, any given point $z\in\mathbb{C}$ can in principle belong to almost any region. In terms of implementation, this choice can be relevant as the expansion in one region may be more accurate than that in another region for a given point. We return to this remark in \S \ref{SsizeReg}.
\end{remark}

All our results are formulated in terms of a particular complex matrix-valued function $R(z) \in \mathbb{C}^{2\times 2}$. In this section we first elaborate briefly on the properties of the function $R(z)$. Next, we introduce auxiliary functions that are needed in the statements of the expansions. Finally, the asymptotic expansions of the polynomials are stated region by region.

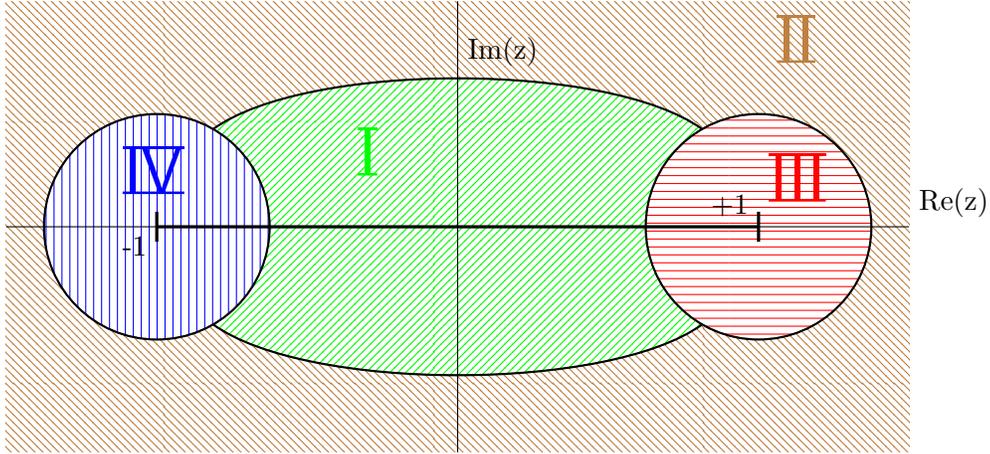
\begin{figure}[t]
\begin{center}
\begin{tikzpicture}
	\fill[pattern=north west lines, pattern color =brown] (0,4) rectangle (12, -2);
	\fill[white] (2.75,2.32) rectangle (9.25, -0.32);
	\draw (2,1) [fill=white] circle [radius=1.5];
	\draw (10,1) [fill=white] circle [radius=1.5];
	\draw[fill=white] (2.75, 2.3) .. controls(4,3.2) and (8, 3.2) .. (9.25, 2.3);
	\draw[fill=white] (2.75, -0.3) .. controls(4,-1.2) and (8, -1.2) .. (9.25, -0.3);

	\fill[pattern=north east lines, pattern color =green] (2.75,2.3) rectangle (9.25, -0.3);
	\draw (2,1) [fill=white] circle [radius=1.5];
	\draw (10,1) [fill=white] circle [radius=1.5];
	\draw (10,1) [thick,pattern=horizontal lines, pattern color =red] circle [radius=1.5];
	\draw (2,1) [thick,pattern=vertical lines, pattern color =blue] circle [radius=1.5];
	\draw[thick,pattern=north east lines, pattern color =green] (2.75, 2.3) .. controls(4,3.2) and (8, 3.2) .. (9.25, 2.3);
	\draw[thick,pattern=north east lines, pattern color =green] (2.75, -0.3) .. controls(4,-1.2) and (8, -1.2) .. (9.25, -0.3);

	\draw[very thick] (2,1)--(10,1);
	\draw (6,-2)--(6,4);
	\draw (6,3) node[above right] {Im(z)};
	\draw (0,1)--(12,1);
	\draw (12,1) node[above right] {Re(z)};
	\draw (2,1) node[below left] {-1};
	\draw (10,1) node[above left] {+1};
	\draw[very thick] (2,0.8)--(2,1.2);
	\draw[very thick] (10,0.8)--(10,1.2);
	\draw[green] (4.8,2) node {\Huge{\Ri } };
	\draw[brown] (10.5,3.5) node {\Huge{\Ro } };
	\draw (10,1.2) node[red,above right] {\Huge{\Rright } };
	\draw (2.5,1.3) node[above left,blue] {\Huge{\Rleft } };
\end{tikzpicture} 
\end{center}
\caption{Regions of the complex plane in which the polynomials have different asymptotic expansions: the lens (\Ri), the outer region (\Ro) and the right and left disks (\Rright~and \Rleft).}
\label{Fregions}
\end{figure}

\subsection{The function $R(z)$ in the complex plane} \label{SfctR}

The function $R(z)$ is a $2\times 2$ matrix complex--valued function, that satisfies the following properties:
\begin{enumerate}
	\item $R(z)$ is analytic (entrywise) in $\mathbb{C}\setminus\Sigma_R$, where the contour $\Sigma_R$ is depicted in Figure \ref{fig_GR}. This contour consists of the boundaries of the regions in Figure \ref{Fregions}. In each piece of $\Sigma_R$ minus the self intersection points, the function $R(z)$ admits boundary values $R_{\pm}(z)$, where the plus (minus) sign corresponds to the left (right) side with the given orientation. 
	\item As $n\to\infty$, the function $R(z)$ admits an asymptotic expansion of the form
	\begin{equation}\label{asympRn}
		R(z)\sim I+\sum_{k=1}^{\infty} \frac{R_k(z)}{n^k}, \qquad n \rightarrow \infty
	\end{equation}
	which is valid uniformly for $z\in\mathbb{C}\setminus(\partial U_{\delta}\cup\partial \tilde{U}_{\delta})$. Here, $U_{\delta}$ and $\tilde{U}_{\delta}$ are the right and left disks respectively.
	\item The coefficients $R_k(z)$ in the previous expansion are analytic functions of $z$ in $\mathbb{C}\setminus(\partial U_{\delta}\cup\partial \tilde{U}_{\delta})$. 
\	\item $R_k(z) = \mathcal{O}(1/z)$ for $z \to \infty$.
\end{enumerate}

\begin{figure}
 \centerline{\includegraphics[width=90mm,height=45mm]{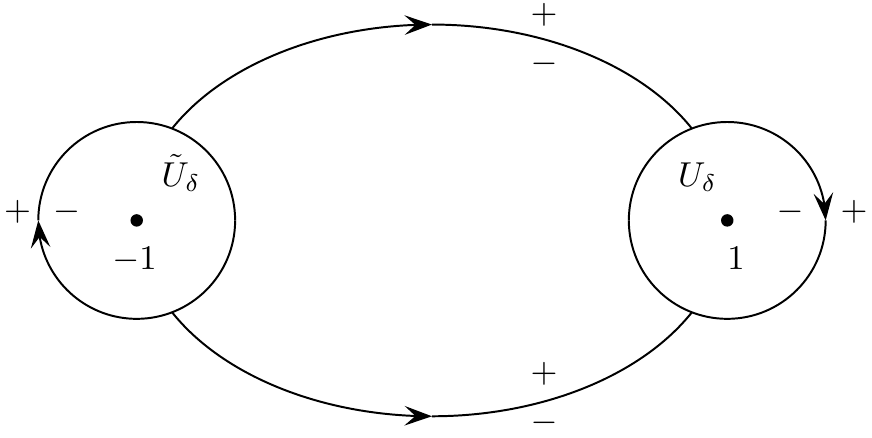}}
\caption{The system of contours $\Sigma_R$ consists of the boundaries of the regions in Figure \ref{Fregions}. The contours are oriented as shown.}
\label{fig_GR}
\end{figure}

It is crucial to note that the coefficients $R_k(z)$ depend on $z$ and are given by different expressions inside and outside of the disks $U_{\delta}$ and $\tilde{U}_{\delta}$. We will write $R_k^{\R}(z)$ and $R_k^{\L}(z)$ to refer to the coefficients for $z$ in the interior of $U_{\delta}$ and $\tilde{U}_{\delta}$ respectively, and $R_k^{\O}(z)$ to indicate the coefficients for $z$ outside these two disks. 

Because of the above properties, this matrix $R(z)$ is close to the identity as $n\to\infty$, uniformly in $z$. Thus, for the leading order behaviour of the expansions, one may simply substitute $R(z)=I$, the $2 \times 2$ identity matrix. Higher-order expansions are obtained automatically by determining the coefficients $R_k(z)$ in formula \eqref{asympRn} for $k\geq 1$ and then plugging in an asymptotic expansion for $R(z)$ that will be derived in \S \ref{S:higherorderterms} and \S\ref{S:simplifications}. The first four terms are given explicitly in \S  \ref{APP:explicit}.

\subsection{Auxiliary functions}\label{ss:auxiliary}

In order to formulate the asymptotic expansions in the different regions of the complex plane explained before, we need some auxiliary functions. In this section we state their definitions. Additional comments about the computation of these functions are given in \S\ref{Sniai}.

The global behavior of $\pi_n(z)$, away from the interval $[-1,1]$ is governed by the Szeg\H{o} function $D(z)$ corresponding to the weight function $w(z)$, which is an analytic function for $z\in\mathbb{C}\setminus [-1,1]$. In our case, because of the form of $w(x)$, see \eqref{wx}, we can write
\begin{equation}\label{Dz}
 D(z)=\frac{(z-1)^{\alpha/2}(z+1)^{\beta/2}}{\varphi(z)^{(\alpha+\beta)/2}}
\exp\left(\frac{(z^2-1)^{1/2}}{4\pi i}\oint_{\tilde{\gamma}} \frac{\log h(\zeta)}{(\zeta^2-1)^{1/2}}\frac{d\zeta}{\zeta-z} \right),
\end{equation}
where $\tilde{\gamma}$ is a closed contour in the complex plane that encircles the interval $[-1,1]$ once in the positive direction but not the point $z$, see \cite[\S 1.1]{KMcLVAV}. 
In this neighbourhood, $h$ needs to have a positive real part and we take the branch of the logarithm that is real on $[-1,1]$.

We also use the function
\begin{equation} \label{E:phiarccos}
  \varphi(z) = z + (z^2-1)^{1/2}, 
\end{equation}
which is a conformal map from $\mathbb{C}\setminus[-1,1]$ onto the complex plane outside the unit circle. Note that we choose the branch cut of the square root on $[-1,1]$. An alternative expression for $\varphi(z)$ can be given in terms of the arccosine function, using the standard definition with a cut on $(-\infty,-1]\cup[1,\infty)$, see \cite[\S4.23.24 \& \S4.23.25]{DLMF}:
\begin{equation} \label{Ephiarccos}
 \varphi(z)=\exp(i\theta(z)\arccos(z)), \qquad 
\theta(z)=\begin{cases}
           1,& \quad \arg(z-1)>0,\\
           -1,& \quad \arg(z-1)\leq 0.
          \end{cases}
\end{equation}
The function $\theta(z)$ corresponds to $\textrm{sgn(Im}\,z)$ in $\mathbb{C}\setminus\mathbb{R}$, and on the real axis it is equal to $-\textrm{sgn}(z-1)$. This function primarily serves a practical purpose, namely to ease the implementation of the branch cuts of several functions in this paper.

We can rewrite $D(z)$ as follows, using standard branch cuts:
\begin{equation}
	D(z) = w(z)^{1/2}\exp(-i\theta(z)\psi(z)). \nonumber
\end{equation}
Here, the phase function $\psi(z)$ captures the oscillatory behavior of the polynomials $\pi_n(z)$ on the interval $[-1,1]$. In the complex plane, it is given by 
\begin{equation}\label{EpsiContz}
	\psi(z) = \frac{\alpha+\beta}{2}\arccos z-\frac{\alpha\pi}{2}
+\frac{(1-z^2)^{1/2}}{4\pi i}\oint_{\gamma} \frac{\log h(\zeta)}{(\zeta^2-1)^{1/2}}\frac{d\zeta}{\zeta-z}, 
\end{equation}
where $\gamma$ does encircle the point $z$ now as well, see \cite[(3.9)]{KV}. The expansions in the next section are formulated in terms of $\psi(z)$, rather than in terms of the Szeg\H{o} function $D(z)$ itself. Note that $\psi(z)$ depends on the analytic function $h(z)$ in the generalized Jacobi weight function through the contour integral in \eqref{EpsiContz}.

Observe that this same contour integral,
\begin{equation}
m(z)= \frac{1}{2\pi i} \oint_{\gamma} \frac{\log h(\zeta)}{(\zeta^2-1)^{1/2}}\frac{d\zeta}{\zeta-z}, \nonumber
\end{equation}
is an analytic function of the variable $z$ in $\mathbb{C}\setminus\gamma$, in particular at $z=\pm 1$. Therefore, we can expand it in power series 
\begin{equation}\label{seriesm}
m(z) \sim \sum_{n=0}^{\infty} c_n (z-1)^n, \qquad
m(z) \sim \sum_{n=0}^{\infty} d_n (z+1)^n, 
\end{equation}
and apply Cauchy's integral formula to obtain
\begin{equation}\label{cndn}
\begin{aligned}
c_n&=\frac{1}{2\pi i} \oint_{\gamma} \frac{\log h(\zeta)}{(\zeta^2-1)^{1/2}}\frac{d\zeta}{(\zeta-1)^{n+1}}, \\
d_n&=\frac{1}{2\pi i} \oint_{\gamma} \frac{\log h(\zeta)}{(\zeta^2-1)^{1/2}}\frac{d\zeta}{(\zeta+1)^{n+1}}, 
\end{aligned}
\end{equation}
for $n\geq 0$. These coefficients $c_n$ and $d_n$ were introduced in \cite[Lemma 6.4 \& 6.6]{KMcLVAV} and are used to construct $R(z)$, see \S\ref{ss:delta_k}. The convergence properties of \eqref{seriesm} depend naturally on the behavior of the function $h$ in the complex plane. We note the following symmetry relation: if $h(-\zeta) = h(\zeta)$, then $d_n = (-1)^{n+1}c_n$.

Finally, we will need the limit of the Szeg\H{o} function $D(z)$:
\begin{equation} \label{DinfCont}
	D_{\infty} = \lim_{z\to\infty} D(z) = 2^{-\frac{\alpha+\beta}{2}}\exp\left( \frac{1}{4\pi i}\oint_\gamma\frac{\log h(\zeta)}{(\zeta^2-1)^{1/2}} d\zeta \right). 
\end{equation}

All contour integrals in this section can be computed either symbolically using residue calculus or numerically using trapezoidal rules in the complex plane, as explained in \S\ref{Scontour}.

\subsection{Asymptotics of monic orthogonal polynomials $\pi_n(z)$}

\subsubsection{Monic polynomials in the lens \Ri } \label{Sint}

Putting together the consecutive transformations in \cite{KMcLVAV} for $z \in$ \Ri, we obtain 
\begin{equation}
	\pi_n(z) = \frac{2^{1/2-n}}{\sqrt{w(z)}(1-z^2)^{1/4} } \begin{pmatrix} 1 \\ 0 \end{pmatrix}^T R^{\O}(z) \begin{pmatrix} D_\infty \cos \left( \lambda_{n,1}(z) \right) \\ -i/D_\infty \cos\left( \lambda_{n,2}(z) \right) \end{pmatrix} \label{EpiInt}
\end{equation}
with branch cuts implemented as in \S \ref{ss:branches} and the following phase functions:
\begin{align}
	\lambda_{n,1}(z) & = (n + 1/2)\arccos z + \psi(z) -\pi/4, \label{Elambda} \\
	\lambda_{n,2}(z) & = (n - 1/2)\arccos z + \psi(z) -\pi/4, \nonumber
\end{align}
with $\psi(z)$ given by \eqref{EpsiContz}, and $D_{\infty}$ as in \eqref{DinfCont}. 

In particular, this expansion shows that the orthogonal polynomial has cosine-like behaviour in the interior of the interval $(-1,1)$, with a frequency that depends on $n$. One can show that the expression is actually valid in all of region \Ri~by analytic continuation of all underlying functions.

\subsubsection{Monic polynomials in the outer region \Ro } \label{Sout}

For $z \in$ \Ro, the asymptotic expansion is
\begin{equation}
	\pi_n(z) = \frac{2^{-1/2-n}}{\sqrt{w(z)}(1-z^2)^{1/4} } \begin{pmatrix} 1 \\ 0 \end{pmatrix}^T R^{\O}(z) \begin{pmatrix} D_\infty \exp\left( i\theta(z) \lambda_{n,1}(z) \right) \\ -i/D_\infty \exp \left( i\theta(z)\lambda_{n,2}(z) \right) \end{pmatrix}. \label{EpiOut}
\end{equation} 
Note that this formulation differs from \cite[\S 1.2 \& (9.2)]{KMcLVAV}: it is numerically more stable, since it avoids raising $\varphi(z)$ to some power and allows re-use of the same contour integrals. 

However, when $\gamma$ would contain points where the analytic continuation of $\log{h(z)}$ is not guaranteed, one could use the formulas in \cite[\S 1.2 \& (9.2)]{KMcLVAV}. It is also possible to define $\tilde{\lambda}_{n,1}(z)$ and $\tilde{\lambda}_{n,2}(z)$ with the only difference that the contour integral in $\psi(z)$ is taken only around [-1,1] (not $z$), and use 
\begin{equation} 
	\pi_n(z) = \frac{2^{-1/2-n}\begin{pmatrix} 1 \\ 0 \end{pmatrix}^T R^{\O}(z)}{(z-1)^{\alpha/2}(z+1)^{\beta/2}(1-z^2)^{1/4} }  \begin{pmatrix} D_\infty \exp\left( i\theta(z) \tilde{\lambda}_{n,1}(z) \right) \\ -i/D_\infty \exp \left( i\theta(z)\tilde{\lambda}_{n,2}(z) \right) \end{pmatrix} \nonumber
\end{equation}

The polynomials behave like complex exponentials in region \Ro. Note that away from the interval \eqref{EpiOut} is exponentially close to \eqref{EpiInt} as $n \rightarrow \infty$. One may think of the polynomials as being asymptotically a sum of two complex exponentials. In region \Ri~the exponentials are of comparable size and they combine into a cosine. In region \Ro, one of the exponentials dominates the other. Hence, the asymptotic expression simplifies.

\subsubsection{Monic polynomials in the right disk \Rright } \label{Sboun}

A formula for $x \in (1-\delta,1]$ is given in \cite[\S 10]{KMcLVAV} and \cite[(2.27)]{KV}. One can extend this result to $z \in \Rright$: 
\begin{equation}
	\pi_n(z) =\frac{\sqrt{\pi n \arccos z} }{2^n\sqrt{w(z)} (1-z^2)^{1/4} } \begin{pmatrix} 1 \\ 0 \end{pmatrix}^TR^{\R}(z) B(z), \label{Epiboun} 
\end{equation}
with
\begin{equation*}
	B(z)  = \begin{pmatrix} D_\infty \left( \cos \left( \zeta_1(z) \right) J_\alpha(n\arccos z) + \sin \left(\zeta_1(z) \right) J_\alpha'(n\arccos{z}) \right) \\ \frac{-i}{D_\infty} \left( \cos \left( \zeta_2(z) \right) J_\alpha(n\arccos z) + \sin \left(\zeta_2(z) \right) J_\alpha '(n\arccos{z}) \right) \end{pmatrix},
\end{equation*}
and
\[
	\zeta_{1,2}(z) = \psi(z)+\frac{\alpha\pi}{2} \pm \frac{1}{2}\arccos z,
\]
where the $+$ sign corresponds to $\zeta_1(z)$ and the $-$ sign to $\zeta_2(z)$.

The polynomials behave like a Bessel function near the right endpoint $x=1$. This is typical asymptotic behaviour near a so-called `hard edge', in the language of random matrix theory. The order of the Bessel function $\alpha$ corresponds to the order of the algebraic singularity of the weight function through the factor $(1-x)^\alpha$. It is unaffected by the presence of the analytic factor $h(x)$.

\subsubsection{Monic polynomials in the left disk \Rleft } \label{Slboun}

For $z \in $ \Rleft, which includes the left part of the interval, we obtain similarly
\begin{equation}
	\pi_n(z) =\frac{\sqrt{\pi n\arccos(-z)}}{(-2)^{n}\sqrt{w(z)} (1-z^2)^{1/4} } \begin{pmatrix} 1 \\ 0 \end{pmatrix}^T R^{\L}(z) B(z) \label{Epilboun}
\end{equation}
with
\begin{equation*}
	B(z) = \begin{pmatrix} D_\infty \left( \sin(\mu_1(z) ) J_\beta\left(n\arccos(-z)\right) + \cos(\mu_1(z) ) J_\beta' \left(n\arccos(-z)\right) \right) \\ \frac{-i}{D_\infty} \left( \sin(\mu_2(z) ) J_\beta\left(n\arccos(-z)\right) + \cos(\mu_2(z) ) J_\beta' \left(n\arccos(-z)\right) \right) \end{pmatrix},
\end{equation*}
and
\[
	\mu_{1,2}(z) = \psi(z) -\frac{\beta\pi}{2} \pm \frac{1}{2}\arccos z,
\]

The polynomials behave like a Bessel function of order $\beta$ near the left endpoint $x=-1$. Note that, compared to \eqref{Epiboun}, the roles of $\alpha$ and $\beta$ are interchanged and other symmetries can be identified. We found that it is simpler to construct explicit formulas for the left disk, rather than to infer them from the formulas of the right disk by invoking symmetry.

\subsection{Asymptotics of leading order coefficients}

The asymptotic expansion of the leading coefficients $\gamma_n$ of the orthonormal polynomials $p_n(x) = \gamma_n\pi_n(x)$ is \cite[\S 9.2]{KMcLVAV}
\begin{align} 
	\gamma_n^2 & \sim \frac{2^{2n}}{\pi D_\infty^2} \left(1+2iD_\infty^2\sum_{k = 1}^\infty \left.\frac{U_{k,1}^{\R} + U_{k,1}^{\L} }{(n+1)^k}\right|_{2,1} \right). \label{Egamma}
\end{align}
The quantities $U_{k,1}^{\R/\L}$ are defined and extensively described in \S\ref{S:higherorderterms}. They are the constant $2\times 2$ matrices that multiply $(z \mp 1)^{-1} n^{-k}$ in the expansion for $R(z)$, of which we use the lower left elements here. Explicit expressions for these matrices up to $k=4$ are given in Appendix \ref{APP:explicit}.

\begin{remark}
We do not state asymptotic expansions for the orthonormal polynomials. These can be obtained simply by multiplying the expansion for the monic polynomial with that of the leading order coefficient $\gamma_n$. Common factors can be cancelled to avoid roundoff or overflow and this is included in the implementation.
\end{remark}

\subsection{Asymptotics of recurrence coefficients}
In the three term recurrence relation \eqref{TTRR}, the recurrence coefficients have the following large $n$ asymptotic expansion \cite[\S 9.3]{KMcLVAV}
\begin{equation}\nonumber 
   \alpha_n \sim -\sum_{k=1}^\infty \left( \left.\frac{U_{k,1}^{\R} + U_{k,1}^{\L} }{(n+1)^k}\right|_{1,1} +\left.\frac{U_{k,1}^{\R} + U_{k,1}^{\L} }{n^k}\right|_{2,2} \right) 
   \end{equation}
and
\begin{equation}\nonumber 
   \beta_n \sim \left(\frac{1}{2iD_\infty^2} + \sum_{k=1}^\infty \left.\frac{U_{k,1}^{\R} + U_{k,1}^{\L} }{n^k}\right|_{2,1} \right) \left(\frac{-D_\infty^2}{2i} + \sum_{k=1}^\infty \left.\frac{U_{k,1}^{\R} + U_{k,1}^{\L} }{n^k}\right|_{1,2} \right).
\end{equation}
The quantities $U_{k,1}^{\R/\L}$ in these expressions are the same as those appearing in \eqref{Egamma} above. Following \cite[Theorem 1.10 \& \S 9.3]{KMcLVAV}, we note that the order $1/n$ terms in the previous expressions cancel out. This can be easily checked with the formulas for $U_{k,1}^{\R/\L}$ given in the Appendix \ref{APP:explicit}, and gives the estimations
\begin{equation}
\alpha_n=\mathcal{O}(1/n^2), \qquad
\beta_n=\frac{1}{4}+\mathcal{O}(1/n^2), \qquad n\to\infty.\nonumber
\end{equation}

\subsection{Remarks on asymptotic expansions} \label{Sremarks}

The asymptotic expansions are stated in \cite{KMcLVAV} for points $x \in \mathbb{R}$ on the interval. Proofs for the validity of their extension to points $z \in \mathbb{C}$ in a region containing (part of) the interval, as they are stated in this paper, are omitted for the sake of brevity. One has to carefully consider the branch cuts involved, which are discussed in \S\ref{Sniai}.

For general $\alpha$ and $\beta$, these asymptotic expansions correspond to a relative error of size $\mathcal{O}(n^{-T})$, where $T$ is the number of terms. If $\alpha^2 = 1/4 = \beta^2$, all higher-order terms are zero ($R(z) = I$) and \eqref{EpiInt}, \eqref{Epiboun} and \eqref{Epilboun} coincide \cite[Rem. 1.14]{KMcLVAV}. In that case, the leading order term of \eqref{EpiInt} gives already exponential accuracy \cite[Rem. 1.5 \& 1.11]{KMcLVAV} and the function $h(x)$ is taken into account only in the definition of $\psi(x)$. If in addition $h(x) = 1$, then we obtain the explicit form of the Chebyshev polynomials.

Although technical, these expansions can readily be differentiated and this is included in the implementation. In \cite{TTOGauss}, derivatives were used as part of the computation of the points and weights of (generalized) Gauss-Hermite quadrature on the real line. There, in the generalized case, the polynomials were evaluated by a numerical solution of the corresponding Riemann-Hilbert problem. As we mentioned in the introduction, the expansions in this paper may be used to compute Gaussian quadrature rules on $[-1,1]$. In the implementation we have included a test script for this computation, based on a Newton method similar to that of \cite{TTOGauss}. This paper affirmatively answers the question raised in the conclusions of \cite{TTOGauss}, whether high-order asymptotic expansions can be effectively derived from a Riemann-Hilbert formulation. An extension to Laguerre weights and generalized Laguerre weights is under current investigation.

Finally, we note that the asymptotic expansions \eqref{EpiInt} -- \eqref{Epilboun} also lend themselves to a cosine transform in order to improve accuracy near the endpoints (see, e.g., \cite{bogaert}). Accuracy may be lost in expressions involving $(1-x^2)$ when $x$ is close to $\pm 1$, due to cancellation. One may substitute $\arccos x = \theta$, and then we have for example that $(1-x^2)^{-1/4} = (\sin\theta)^{-1/2}$, which is numerically stable for $\theta$ close to $0$. In our implementation, we have also included series expansions around the endpoints. There, the particular singularity $(1-x^2)^{-1/4}$ is cancelled analytically with other terms, as well as the singularity that arises from $w(z)^{-1/2}$.

\section{Computation of higher-order terms}\label{S:higherorderterms}

\subsection{Local jumps for the matrix $R(z)$}\label{ss:localjump}

It follows from \eqref{asympRn} that the matrix-valued function $R(z)$ is close to the identity matrix as $n\to\infty$. In fact, the leading order terms of the expansions in \S\ref{Sasy} are obtained by simply substituting $I$ in the previous expressions.

The jumps of the matrix $R(z)$ across the contour $\Sigma_R$, shown in Figure \ref{fig_GR}, tend to the identity matrix as $n\to\infty$, but we have two different types of behaviour. The first type of jump is exponentially small in $n$,
\begin{equation*}
   R_+(z) = R_-(z)\left(I+\mathcal{O}(e^{-2c n})\right), \qquad c>0,
\end{equation*}
and holds on the lips of the lens--shaped region, which is the boundary between the regions \Ri~and \Ro. On the other hand, on the boundary of the disks around the endpoints we have
\begin{equation}
   R(z) = R^{\R/\L}(z) \left(I+\mathcal{O}\left(\frac{1}{n}\right)\right), \label{Ejump}
\end{equation}
uniformly for $z\in \partial U_{\delta}\cup\partial\tilde{U}_{\delta}$.

The main idea to obtain higher-order terms in the asymptotic expansion for $\pi_n(z)$ is to compute the higher-order terms $R_k(z)$ in \eqref{asympRn}. To this end, we write the jump matrix for $R(z)$ as a perturbation of the identity matrix, $I+\Delta(z)$, i.e. we write \eqref{Ejump} as
\begin{equation}\label{Ejump2}
R(z) = R^{\R/\L}(z)(I+\Delta^{\R/\L}(z)).
\end{equation}
We then consider a full asymptotic expansion in powers of $1/n$ for $\Delta(z)$:
\begin{equation}\nonumber
   \Delta(z) \sim \sum_{k=1}^{\infty} \frac{\Delta_k(z)}{n^k}, \qquad n\to\infty,
\end{equation}
uniformly for $z \in \Sigma_R$. The terms $\Delta_k(z)$ are identically $0$ in $\Sigma_R\setminus (\partial U_{\delta}\cup\partial \tilde{U}_{\delta})$, because the jump of the first type is exponentially close to the identity there. On the boundary of the disks, the terms $\Delta_k(z)$ can be written explicitly as $\Delta_k^{\R/\L}(z)$, as we detail next.

\subsection{The definition of $\Delta_k$}\label{ss:delta_k}

An explicit expression for $\Delta_k(z)$ is known \cite{KMcLVAV}:
\begin{align}
	\Delta_k^{\R}(z)& = \frac{(\alpha,k-1)}{\left(2 \log \varphi(z) \right)^k}D_{\infty}^{\sigma_3} M(z) F^{\R}(z)^{\sigma_3}    \label{EDeltak}    \\
	& \times \begin{pmatrix} \tfrac{(-1)^k}{k}(\alpha^2+\tfrac{k}{2} -\tfrac{1}{4}) & -i\left(k-\tfrac{1}{2}\right) \\[1mm] (-1)^k\left(k-\tfrac{1}{2}\right)i & \tfrac{1}{k}(\alpha^2 +\tfrac{k}{2} -\tfrac{1}{4}) \end{pmatrix}F^{\R}(z)^{-\sigma_3}M(z)^{-1}D_{\infty}^{-\sigma_3}. \nonumber 
\end{align}
Here, we have used the notation $(\alpha,m)$ to denote, for $m > 0$,
\begin{equation}\label{E:bracket}
	(\alpha,m) =\frac{ \prod_{n=1}^{m}(4 \alpha^2-(2n-1)^2) }{2^{2m}m!}.
\end{equation} 
along with $(\alpha,0)=1$. Also, $D_{\infty}$ is given by \eqref{DinfCont}, and 
\begin{align}
	M(z) & = \frac{1}{\sqrt{2}(z^2-1)^{1/4}} \begin{pmatrix} \varphi(z)^{1/2} & i \varphi(z)^{-1/2}      \\  	-i\varphi(z)^{-1/2} & \varphi(z)^{1/2} \end{pmatrix},      \label{Mz}     \\
	& = \frac 12 \begin{pmatrix} \gamma(z) + \gamma(z)^{-1} & -i(\gamma(z) - \gamma(z)^{-1}) \\ i(\gamma(z) - \gamma(z)^{-1})  & \gamma(z) + \gamma(z)^{-1} \end{pmatrix}, \nonumber
\end{align}
with $\gamma(z) = \left(\frac{z-1}{z+1}\right)^{1/4}$ and $\varphi(z)$ given by \eqref{E:phiarccos}. The function $F^{\R}(z)$ is 
\begin{align}\label{Fright}
F^{\R}(z) =  &\exp\left(i\theta(z)\left(\psi(z) +\frac{\alpha\pi i}{2}\right)\right).
\end{align}
This function is analytic in $U_{\delta}\setminus (1-\delta,1]$, and it has an expansion there:
\begin{equation}\label{seriesFcn}
	F^{\R}(z) \sim \varphi(z)^{(\alpha+\beta)/2}\exp\left( \frac{1}{2} (z^2-1)^{1/2}\sum_{n=0}^\infty  c_n(z-1)^n \right), \quad z \in U_\delta,
\end{equation}
with coefficients $c_n$ defined in \S\ref{ss:auxiliary} by \eqref{cndn}. In the implementation, see \S \ref{ss:explicitexpansion}, we use this formula combined with $\varphi(z)$ in terms of the arccosine, see \eqref{Ephiarccos}.

We have also used the standard notation for the third Pauli matrix $\sigma_3$,  
\begin{equation}
\sigma_3=\begin{pmatrix} 1 & 0 \\ 0 &-1\end{pmatrix}, \qquad \textrm{and}\qquad
   f(z)^{\sigma_3} = \begin{pmatrix} f(z) & 0\\ 0 & f(z)^{-1} \end{pmatrix}, \nonumber
\end{equation}
for any function $f(z)\neq 0$.

There are analogous formulas for $\Delta_k^{\L}(z)$ for $z$ near $-1$:
\begin{align}
	\Delta_k^{\L}(z)& = \frac{(\beta,k-1)}{\left(2 \log [-\varphi(z)] \right)^k}D_{\infty}^{\sigma_3} M(z) F^{\L}(z)^{\sigma_3}      \label{DeltaL}     \\
	&  \times 	\begin{pmatrix} \tfrac{(-1)^k}{k}(\beta^2+\tfrac{k}{2} -\tfrac{1}{4}) & i\left(k-\tfrac{1}{2}\right) \\[1mm] (-1)^{k+1}\left(k-\tfrac{1}{2}\right)i & \tfrac{1}{k}(\beta^2 +\tfrac{k}{2} -\tfrac{1}{4}) \end{pmatrix}F^{\L}(z)^{-\sigma_3}M(z)^{-1}D_{\infty}^{-\sigma_3},
\end{align}
with
\begin{align}\label{Fleft}
	F^{\L}(z) = & \exp\left(i\theta(z)\left(\psi(z) -\frac{\beta\pi i}{2}\right)\right),
\end{align}
which is an analytic function in $\tilde{U}_{\delta}\setminus [-1,-1+\delta)$, with
\begin{equation}\label{seriesFdn}
	F^{\L}(z) \sim ( -\varphi(z))^{(\alpha+\beta)/2} \exp\left(\frac{1}{2} (z^2-1)^{1/2}\sum_{n=0}^\infty  d_n(z+1)^n \right), \quad z \in \tilde{U}_\delta 
\end{equation}
and coefficients $d_n$ given by \eqref{cndn}.

It is important to note that by $(z^2-1)^{1/2}$ we mean the analytic branch of the square root that behaves like $z$ as $z \rightarrow \infty$ in any direction. We comment further on the correct implementation of this expression in \S\ref{ss:branches}.

\begin{remark}\label{rem:specialcases}
The special case with $\alpha^2=\beta^2=1/4$, mentioned before in \S\ref{Sremarks}, follows from the form of the matrices in \eqref{EDeltak} and the fact that $(\alpha,m)$ and $(\beta,m)$ vanish for any $m \geq 1$. This is easily seen from the definition \eqref{E:bracket} of $(\alpha,m)$. It follows also that $(\alpha,m)$ vanishes whenever $\alpha$ is a half-integer, once $m$ surpasses a certain maximal value. This implies that in these cases most coefficients $\Delta_k$ vanish identically. This simplifies computations somewhat, but not greatly, since $R(z)$ does not necessarily have a finite number of terms nor poles. Such cases include Gegenbauer polynomials or the closely related spherical polynomials, for example, of the kind employed in the spectral method presented in \cite{olverTownsend}.
\end{remark}

\subsection{Recursive computation of $R_k(z)$} \label{ScompU}

We recall expression \eqref{asympRn} for convenience:
\begin{equation} \nonumber
 R(z)\sim I+\sum_{k=1}^{\infty} \frac{R_k(z)}{n^k}, \qquad n \rightarrow \infty, \qquad z\in\mathbb{C}\setminus( U_{\delta}\cup{\tilde{U}_{\delta}}).
\end{equation}
The function $R(z)$ is analytic in the regions \Ri, \Ro, \Rright~and \Rleft, but has jumps across the contour $\Sigma_R$.
Recall that we write $R_k^{\R/\L}(z)$ to refer to the coefficients in the interior of the right/left disks, and $R_k^{\O}(z)$ for the coefficients outside of the disks.

By expanding the jump relation \eqref{Ejump2} and collecting the terms with equal order in $n$, we obtain a link between the terms $R_k(z)$ in the expansion \eqref{asympRn} and the $\Delta_k$. For $k\geq 1$, we have 
\begin{equation}\label{RHPforRk}
 R_k^{\O}(z)=R_k^{\R/\L}(z)+\sum_{j=1}^k R^{\R/\L}_{k-j}(z)\Delta^{\R/\L}_j(z), \quad z\in\partial U_{\delta}\cup\partial\tilde{U}_{\delta},
\end{equation}
with $R_0^{\R/\L}(z)=I$, cf. \cite[(8.12)]{KMcLVAV}. 

This is an \emph{additive Riemann-Hilbert problem} for the $R_k(z)$. We are looking for a solution to this problem recursively for each value of $k$, and at each step of the recursion, the additive jump involves the solutions of previous problems, i.e. $R_j$ with $j < k$.

All quantities involved are meromorphic functions in $z$. It should be noted that the functions $\Delta^{\R}_j(z)$ and $R_j(z)$ may have poles at $z=1$, but $R_j^{\R}(z)$ may not, since the latter is analytic in the right disk. Thus, one can solve the additive Riemann-Hilbert problem as follows:
\begin{itemize}
 \item Expand the sum in \eqref{RHPforRk} in a Laurent series around $z=\pm 1$.
 \item Define $R_k^{\O}(z)$ as the sum of all the terms containing strictly negative powers of $z\mp 1$. Since $R_k(z) = \mathcal{O}(1/z)$ as $z \rightarrow \infty$, positive powers of $z\mp 1$ do not contribute to $R_k^{\O}(z)$.
 \item Define $R^{\R}_k(z)$ as the remainder after subtracting those poles.
\end{itemize}
This construction ensures that $R_k^{\O}$ is analytic outside the disk, $R^{\R}_k$ is analytic inside and \eqref{RHPforRk} holds, as required.

A useful piece of information is conveyed by \cite[Lemma 8.2]{KMcLVAV}: for any $k\geq 1$, the functions $\Delta^{\R/\L}_k(z)$ have a pole at $z=\pm 1$ of order at most $\lfloor(k+1)/2\rfloor=\lceil k/2 \rceil$. Thus, we may write
\begin{equation}\label{Vkm}
  \Delta_k^{\R/\L}(z) \sim \sum_{m=-\lceil k/2 \rceil }^\infty V_{k,m}^{\R/\L} (z\mp 1)^m,
\end{equation}
with coefficients $V^{\R/\L}_{k,m}$ that can be computed explicitly by expanding \eqref{EDeltak} around $z=\pm 1$. It follows that the Laurent expansion of the sum in \eqref{RHPforRk} has a principal part of the same order, which we may write as
\begin{equation} \nonumber 
	\sum_{j=1}^k R_{k-j}^{\R/\L}(z)\Delta_j^{\R/\L}(z) = \sum_{m=1}^{\lceil k/2 \rceil} \frac{U_{k,m}^{\R/\L} }{(z\mp 1)^m}  +\mathcal{O}(1), \quad z \rightarrow \pm 1.
\end{equation}
This expansion defines the $U^{\R/\L}_{k,m}$ coefficients that appear in the asymptotic expansions of the orthogonal polynomials.

\begin{remark}
We have described the coefficients of $R_k^{\O}$ at a sightly greater level of generality than in \cite{KMcLVAV}. In the notation of \cite{KMcLVAV}, the first coefficients are 
$$
\begin{aligned}
A^{(1)} &= U_{1,1}^{\R}, \quad
B^{(1)} = U_{1,1}^{\L}, \quad
A^{(2)} = U_{2,1}^{\R}, \quad
B^{(2)} = U_{2,1}^{\L},\\
A^{(3)} &= U_{3,1}^{\R}, \quad
B^{(3)} = U_{3,1}^{\L}, \quad
C^{(3)} = U_{3,2}^{\R}, \quad
D^{(3)} = U_{3,2}^{\L}.
\end{aligned}
$$
\end{remark} 

The construction outlined above yields
\begin{equation}\label{ERpl}
	R_{k}^{\O}(z) =\sum_{m=1}^{\lceil k/2 \rceil } \left(\frac{U_{k,m}^{\R} }{(z-1)^m} + \frac{U_{k,m}^{\L} }{(z+1)^m} \right), \qquad z\in\mathbb{C}\setminus(U_{\delta}\cup\tilde{U}_{\delta}).
\end{equation}
At the same time, since $R^{\R/\L}_k(z)$ is analytic in $U_{\delta}$ (respectively $\tilde{U}_{\delta}$), it has a local power series expansion
\begin{equation}
	R^{\R/\L}_k(z) \sim \sum_{n=0}^\infty Q_{k,n}^{\R/\L} (z\mp 1)^n, \label{ERrlseries}
\end{equation}
with some coefficients $Q^{\R/\L}_{k,n}$ that can be determined as well. The three sets of coefficients are necessarily related. It follows from the additive jump relation \eqref{RHPforRk}, by expanding around $z=\pm 1$ and comparing equal powers, that they satisfy the identities:
\begin{align}
	U_{k,m}^{\R/\L} = & \hspace{1.5mm} V_{k,-m}^{\R/\L} + \sum_{j=1}^{k-1}\sum_{l=0}^{\lceil j/2 \rceil -m} Q_{k-j,l}^{\R/\L} V_{j,-m-l}^{\R/\L} \label{EUpole}  \\
	Q_{k,n}^{\R/\L} = & \frac{1}{n!}\left(\sum_{i=1}^{\lceil k/2 \rceil} (-i-n+1)_n (\pm 2)^{-i-n} U_{k,i}^{\L/\R} \right)  \label{EpRR} \\
	& -V_{k,n}^{\R/\L} -\sum_{j=1}^{k-1}\sum_{l=0}^{\lceil j/2 \rceil +n} Q_{k-j,l}^{\R/\L} V_{j,n-l}^{\R/\L}, \nonumber
\end{align}
where $+2$ corresponds to $Q_{k,n}^{\R}$ and $-2$ to $Q_{k,n}^{\L}$. Observe that the roles of the right and left superscripts are sometimes interchanged. Also, in the last expressions we have used the notation
\[
(n)_m=n(n+1)\cdots(n+m-1)
\]
to denote the Pochhammer symbol.

A possible approach to compute higher-order terms is to implement a symbolic computation of power series expansions around $z=\pm 1$, and to combine \eqref{RHPforRk}, \eqref{Vkm} and \eqref{ERrlseries} in order to obtain the coefficients $U^{\R/\L}_{k,m}$. However, this procedure turns out to be extremely inefficient symbolically for high-order terms, because many lengthy expressions are constructed and manipulated. Using the relationships \eqref{EUpole} and \eqref{EpRR} improves this situation, but in the following section we explore an alternative way to compute the matrices $U_{k,m}^{\R/\L}$ directly and more efficiently.

\section{Simplifications and explicit formulas}\label{S:simplifications}

The crucial formula in \S\ref{S:higherorderterms} is the jump relation \eqref{RHPforRk}, from which $R_k$ can be determined recursively. In this section, we rewrite the jump relation as \eqref{E:simplifiedjump} below, in such a way that the computation of higher-order terms is significantly accelerated. We also establish explicit formulas for the expansions of all quantities involved, such that higher-order terms can be computed fully numerically, without having to resort to a symbolic computation package.

\subsection{Simplifications} \label{Ssimpl}

In the computations outlined in the previous section, some combinations of $\Delta_k(z)$'s simplify or cancel. This can be used to speed up the computation of $R_k^{\R/\L}(z)$ considerably. We start by writing the jump relation \eqref{RHPforRk} using the coefficients $R_{k-m}(z)$ instead of $R^{\R/\L}_{k-m}(z)$.

\begin{proposition} The jump relation \eqref{RHPforRk} can be written as follows:
\begin{equation}\label{E:simplifiedjump}
	R^{\R/\L}_k(z) = R_k^{\O}(z) - \sum_{m=1}^k R_{k-m}^{\O}(z)s^{\R/\L}_m(z)
\end{equation}
with $R_{0}^{\R/\L}(z) = I$ and with
\begin{equation}\label{EcheckS}
   s^{\R/\L}_m(z) = \Delta^{\R/\L}_m(z) - \sum_{j=1}^{m-1} s^{\R/\L}_j(z)\Delta^{\R/\L}_{m-j}(z).
\end{equation} 
\end{proposition}

\begin{proof}
To prove the result, we proceed by induction for the case of the right disk. The base case $k=1$ is trivial and we assume that \eqref{E:simplifiedjump} holds until $k-1$. Since $k-m \leq k-1$, we can substitute the right hand side of \eqref{E:simplifiedjump} for $R_{k-m}^{\R}(z)$ in \eqref{RHPforRk}. This yields
\begin{align*}
	&R_{k}^{\O}(z) - R_{k}^{\R}(z)  = \sum_{m=1}^k \left( R_{k-m}^{\O}(z) - \sum_{n=1}^{k-m} R_{k-m-n}^{\O}(z)s^{\R}_n(z) \right) \Delta_m^{\R}(z) \\
	      & = \sum_{m=1}^k R_{k-m}^{\O}(z)\Delta^{\R}_m(z) 
	     - \sum_{m=1}^k \sum_{n=1}^{k-m} R_{k-m-n}^{\O}(z)s^{\R}_n(z)\Delta^{\R}_{m}(z).
\end{align*}
\
The second sum can be rewritten using a change of variables $\ell = m+n$:
\begin{align*}
 	&\sum_{m=1}^k \sum_{n=1}^{k-m} R_{k-m-n}^{\O}(z)s^{\R}_n(z)\Delta^{\R}_{m}(z) = \\
 	&=\sum_{\ell=2}^k R_{k-\ell}^{\O}(z)\sum_{n=1}^{\ell-1}s^{\R}_n(z)\Delta^{\R}_{\ell-n}(z) = \sum_{\ell=1}^k R_{k-\ell}^{\O}(z) \left(-s_\ell^{\R}(z)+\Delta_\ell^{\R}(z) \right),
\end{align*}
where it was possible to add the case $\ell=1$ because $s_1^{\R}(z)-\Delta_1^{\R}(z)=0$. This proves the result, and the left case is analogous.
\end{proof}

At first sight, \eqref{E:simplifiedjump} is merely rewriting \eqref{RHPforRk}, but this formulation has two essential advantages:
\begin{itemize}
 \item The jump term in \eqref{E:simplifiedjump} is written in terms of $R_{k-m}^{\O}$ rather than $R^{\R}_{k-m}$, and the former has a simple and non-recursive expression \eqref{ERpl}.
 \item The definition of the coefficients $s_m^{\R/\L}$ can be greatly simplified to a non-recursive expression too, involving just the $\Delta_k$'s. 
\end{itemize}

More precisely, we have the following result:

\begin{proposition} \label{Tsimpl} The terms $s^{\R/\L}_m(z)$ defined by \eqref{EcheckS} satisfy 
\begin{equation}
   s^{\R/\L}_m(z) = \Delta^{\R/\L}_m(z) \nonumber
\end{equation}
for odd $m$ and 
\begin{align}
   s^{\R}_m(z) = & \hspace{1.5mm} \Delta^{\R}_m(z) -\frac{4\alpha^2+2m-1}{\ln(\varphi(z) )^m}\frac{(\alpha,m-1)}{2^{m+1}m} I, \nonumber \\
   s^{\L}_m(z) = & \hspace{1.5mm} \Delta^{\L}_m(z) -\frac{4\beta^2+2m-1}{\ln(-\varphi(z) )^m}\frac{(\beta,m-1)}{2^{m+1}m} I \nonumber
\end{align}
for even $m$, with $(\alpha,m)$ defined by \eqref{E:bracket}.
\end{proposition}
This can be proven again by mathematical induction, see the proof in Appendix \ref{APP:proof_of_simplifications}.

\subsection{Precomputing a series expansion for $s_k^{\R/\L}(z)$} \label{ss:explicitexpansion}

The recursive procedure in \S\ref{ScompU} relied on series expansions. 
Symbolic manipulation of series, which might not be available or slow, can be avoided by deriving explicit formulas for the expansion. In this section, we construct the expansion for the functions $s_k^{\R/\L}(z)$ in order to find the $U_{k,m}^{\R/\L}$. In view of Proposition \ref{Tsimpl}, this is equivalent to deriving the expansion of the $\Delta_k$'s. In turn, this amounts to deriving expansions for all quantities appearing in their (lengthy) definition in \S\ref{ss:delta_k} and combining these expansions through convolutions to obtain the final result. This is conceptually straightforward, but laborious in practice. In this section, we supply rather many technical details, as great care has to be taken with signs and branch cuts, which can only be achieved with thorough understanding of the methodology of \cite{KMcLVAV}.

We want to compute the coefficients $W_{k,m}$ in
\begin{equation} \label{Wkm}
	s_k^{\R/\L}(z) \sim \sum_{m=-\lceil k/2 \rceil}^\infty W_{k,m}^{\R/\L} (z\mp 1)^m, \qquad z\to\pm 1.
\end{equation}

We proceed by detailing the expansion of the quantities appearing in definition \eqref{EDeltak} and afterwards, one by one. We observe that one can write \eqref{EDeltak} as
\begin{equation}\label{Deltak_Tk}
	 \Delta_k^{\R/\L}(z)= \frac{(q,k-1)}{\left(2 \log [\pm \varphi(z) ] \right)^k} D_\infty^{\sigma_3} G_k(z)D_\infty^{-\sigma_3},
\end{equation}	 
where $q=\alpha$ for the right disk and $q=\beta$ for the left disk. Here, the $\pm$ signs always correspond to the right/left endpoint. The function $G_k(z)$ in \eqref{Deltak_Tk} can be given in terms of $M(z)$, see \eqref{Mz}, and $F^{\R/\L}(z)$, see \eqref{Fright} and \eqref{Fleft}. Omitting superscripts for brevity, we have
\begin{equation}
	G_k(z)= M(z)F(z)^{\sigma_3} \begin{pmatrix} (-1)^k a & b \\ (-1)^{k+1} b & a \end{pmatrix} F(z)^{-\sigma_3}M^{-1}(z) \label{Gk}
\end{equation}
with
\[
a = \frac{1}{k}\left(q^2 + \frac{k}{2}-\frac{1}{4}\right) , 
\qquad 
b = \mp \left(k-\frac{1}{2}\right)i.
\]
Working out the multiplication of the matrices for odd and even $k$, we obtain
\[
G_k^{\text{odd}}(z)=\frac{1}{(z^2-1)^{1/2}} \left[ \begin{pmatrix} -az & {ia}\\ ia & az \end{pmatrix} +ib \begin{pmatrix} \cos(y_{\alpha+\beta}) & -i \cos(y_{\alpha+\beta+1})\\ -i \cos(y_{\alpha+\beta-1}) & -\cos(y_{\alpha+\beta})\\ \end{pmatrix} \right]
\]
and
\[
G_k^{\text{even}}(z)=\begin{pmatrix}a & 0\\ 0 & a\end{pmatrix}+\frac{b}{(z^2-1)^{1/2}} \begin{pmatrix} -\sin(y_{\alpha+\beta}) & i\sin(y_{\alpha+\beta+1})\\ i\sin(y_{\alpha+\beta-1}) & \sin(y_{\alpha+\beta}) \end{pmatrix}.
\]

Based on \eqref{seriesFcn} and \eqref{seriesFdn}, the functions $y_\gamma=y_{\gamma}(z)$ above are given by
\begin{equation}\label{ygamma}
	y_\gamma \sim -i\gamma\, \log(\pm \varphi(z)) -i(z^2-1)^{1/2}\,\sum_{n=0}^\infty\begin{cases} c_n(z-1)^n \\ d_n (z+1)^n, \end{cases}
\end{equation}
with $\gamma=\alpha+\beta$ or $\gamma=\alpha+\beta\pm 1$.

In order to compute the coefficients $W_{k,m}$ in \eqref{Wkm}, we will expand all the previous functions in power series around $z=\pm 1$. We will use the notation $v=z\mp 1$, with minus (plus) sign for the right (left) disk.

We start with the power $\log(\pm \varphi(z))^{-k}$ in \eqref{Deltak_Tk}: from \eqref{E:phiarccos}, we get
\begin{equation}
 \log\varphi(z)=i\theta(z)\arccos(z), \label{ElogphiTheta}
\end{equation}
and $\log(-\varphi(z))=\log\varphi(z)-\theta(z)\pi i$. Expanding the arccosine as $z\to \pm 1$, we obtain
\begin{equation}\label{logphi}
	\log(\pm\varphi(z)) \sim (\pm 2v)^{1/2}\sum_{n=0}^{\infty} f_n v^n, \qquad f_n=\frac{(\frac 12)_n }{(\mp 2)^{n}n!(1+2n)},
\end{equation}
using the standard Pochhammer symbol $(\frac 12)_n$ and the variable $v$ explained before. We note that the factor $\theta(z)$ in \eqref{ElogphiTheta} is cancelled by the branches of the logarithm and the square root.
Continuing, we have the recursive result
\[
	(\log(\pm \varphi(z) ))^{-1} \sim (\pm 2v)^{-1/2}	\sum_{n=0}^{\infty} g_{1,n} v^{n}, \qquad 	g_{1,n} = \frac{-1}{f_0} \sum_{j=0}^{n-1} g_{1,j} f_{n-j} 
\]
with $g_{1,0} = 1/{f_0} = 1$, and, for $k>1$,
\begin{equation} \label{gkn}
	(\log(\pm \varphi(z)))^{-k} \sim (\pm 2v)^{-k/2} \sum_{n=0}^{\infty} g_{k,n} v^{n} \qquad g_{k,n} = \sum_{l=0}^n g_{k-1,l} g_{1,n-l}.
\end{equation}

In order to expand $\cos(y_{\gamma})$ and $\sin(y_{\gamma})$, we note first that because of \eqref{ygamma} and \eqref{logphi}, we have
\[
y_{\gamma} \sim -i(\pm 2v)^{1/2}\sum_{n=0}^{\infty}\rho_{1,n,\gamma} v^n, \qquad \rho_{1,n,\gamma}  = \gamma f_n \pm \sum_{j=0}^n {\tfrac{1}{2} \choose j} \begin{cases} c_{n-j} 2^{-j} \\ d_{n-j}(-2)^{-j}\end{cases}.
\]

Note that with the standard branch cuts for the powers, $y_{\gamma}$ is real on the interval $[-1,1]$. Then, for $k>1$, 
\begin{align}
	y_\gamma^k & \sim (-i)^{k}(\pm 2v)^{k/2} \sum_{n=0}^\infty \rho_{k,n,\gamma} v^n, \qquad \rho_{k,n,\gamma} = \sum_{l=0}^n \rho_{k-1,l,\gamma} \rho_{1,n-l,\gamma}, \nonumber 
\end{align}
and
$$
\begin{aligned}
\cos y_\gamma& \sim \sum_{n=0}^\infty H_{n,\gamma}^{\text{odd}} v^n =1 +\sum_{n=1}^\infty\left[ \sum_{j=1}^n (\pm 2)^j \frac{\rho_{2j,n-j,\gamma}}{(2j)!}\right]v^n,\\
\sin y_\gamma& \sim -i(\pm 2v)^{1/2}\sum_{n=0}^\infty H_{n,\gamma}^{\text{even}} v^{n}=-i(\pm 2v)^{1/2}\sum_{n=0}^{\infty} \left[\sum_{j=0}^n (\pm 2)^j \frac{\rho_{2j+1,n-j,\gamma} }{(2j+1)!}\right]v^{n}.
\end{aligned}
$$

One more expansion is needed, as we have to divide by $(z^2-1)^{1/2}$. Since
\[
(z^2-1)^{-1/2} \sim (\pm 2v)^{-1/2}\sum_{n=0}^{\infty}{-\frac{1}{2} \choose n} (\pm 2)^{-n} v^n,
\]
we obtain
\[
\begin{aligned}
\frac{\cos(y_{\gamma})}{(z^2-1)^{1/2}} & \sim (\pm 2v)^{-1/2}\sum_{n=0}^{\infty}\left[1+\sum_{j=1}^n {-\frac{1}{2} \choose j} (\pm 2)^{-j}H^{\text{odd}}_{n-j,\gamma}\right]v^n,\\
\frac{\sin(y_{\gamma})}{(z^2-1)^{1/2}} & \sim -i\sum_{n=0}^{\infty}\left[\sum_{j=0}^n {-\frac{1}{2} \choose j} (\pm 2)^{-j}H^{\text{even}}_{n-j,\gamma}\right]v^n.
\end{aligned}
\]


Also, we observe that
\[
\frac{z}{(z^2-1)^{1/2}} \sim (\pm 2v)^{-1/2}\sum_{n=0}^{\infty}(2n+1){n-\frac{3}{2}\choose n}(\mp 2)^{-n} v^n
\]
to complete the computation of $G_k^{\text{odd}}(z)$ and $G_k^{\text{even}}(z)$.

Finally, bearing in mind \eqref{Deltak_Tk} and Proposition \ref{Tsimpl}, we write the coefficients $W_{k,m}$ as follows.

\begin{proposition}
The coefficients $W_{k,m}^{\R/\L}$ in expansion \eqref{Wkm} for the functions $s_k^{\R/\L}(z)$ are given explicitly by 
\begin{align}
         W_{k,m}^{\R/\L} & = \frac{(q,k-1)}{(\pm 2)^{3k/2}} \sum_{j=0}^{m+(k+1)/2} g_{k,j} G_{k,m+(k+1)/2-j}^{\operatorname{odd} }, \nonumber \\
         W_{k,m}^{\R/\L} & = \frac{(q,k-1)}{(\pm 2)^{3k/2}} \left( \frac{-(4q^2+2k-1)g_{k,m+k/2} }{2k}I + \sum_{j=0}^{m+k/2} g_{k,j} G_{k,m+k/2-j}^{\operatorname{even}} \right),  \nonumber
\end{align}
with $q=\alpha$ for the right disk and $q=\beta$ for the left disk. Here, $g_{k,j}$ are defined by \eqref{gkn} as the coefficients in the expansion of $\log(\pm \varphi(z))^{-k}$ around $z=\pm 1$. The coefficient matrices $G_{k,n}^{\operatorname{odd / even}}$ are the expansion coefficients of $G_k(z)$ defined by \eqref{Gk} around $\pm 1$ for odd and even $k$, respectively.
\end{proposition}

\begin{remark} 
One can compute the values $V_{k,m}^{\R/\L}$ directly in a similar way: compare \eqref{Vkm} and \eqref{Wkm}, using the correspondence between these expressions given in Proposition \ref{Tsimpl}.
\end{remark}

Analogous formulas to \eqref{EUpole}--\eqref{EpRR} can be derived, relating the $U_{k,m}^{\R/\L}$, $Q_{k,m}^{\R/\L}$ and $W_{k,m}^{\R/\L}$ values. However, the coefficients $W_{k,m}^{\R/\L}$ can also be used to compute $U_{k,m}^{\R/\L}$ directly, based on \eqref{E:simplifiedjump}. This requires fewer $W_{k,m}^{\R/\L}$ values than \eqref{EUpole}--\eqref{EpRR} uses $V_{k,m}^{\R/\L}$ values. That leads us to the final formula:
\begin{align}
	U_{k,m}^{\R/\L} = & \hspace{1.5mm} W_{k,-m}^{\R/\L} + 
	\sum_{j=1}^{k-1}\sum_{l=\max(m-\lceil j/2\rceil,1) }^{\lceil (k-j)/2 \rceil} U_{k-j,l}^{\R/\L} W_{j,l-m}^{\R/\L} \label{EUW} \\
	 & \hspace{1.5mm} +\sum_{j=1}^{k-1}\sum_{n=0}^{\lceil j/2 \rceil -m}  \left(\sum_{i=1}^{\lceil (k-j)/2 \rceil} \frac{ (1-i-n)_n }{(\pm 2)^{i}} U_{k-j,i}^{\L/\R} \right) \frac{W_{j,-n-m}^{\R/\L} }{(\pm 2)^{n} n!}. \nonumber
\end{align}

\section{Numerical issues and implementation} \label{Sniai}

\subsection{Square roots and other algebraic singularities}\label{ss:branches}

Several multivalued functions appear in the asymptotic expansions of \S\ref{Sasy}, whose implementation in the complex plane deserves some attention. Recall first the mathematical expression for the $\varphi$ function, first introduced in \eqref{E:phiarccos}, which is
\begin{equation*}
\varphi(z)=z+(z^2-1)^{1/2}.
\end{equation*}

This function is understood to be analytic in $\mathbb{C}\setminus[-1,1]$ and to behave 
like $z$ as $z \rightarrow \infty$. This means that $(z^2-1)^{1/2}$ is the analytic continuation of the square root $\sqrt{x^2-1}$, positive for $x>1$, to the complex plane minus the interval $[-1,1]$. Observe that the square root is negative when $z<-1$, on the negative real axis.

This poses a problem in implementation, since the standard branch cut of the square root function in $(z^2-1)^{1/2}$ results in an extra cut on the imaginary axis, because the argument of the square root is real and negative there. This extra cut is avoided when implementing the expression 
\begin{verbatim}
     phi(z) = z+sqrt(z-1)*sqrt(z+1)
\end{verbatim}
using the standard branch cuts.

Similar considerations apply to other multivalued functions such as
\[
 (1-z^2)^{1/2}, \quad (1-z^2)^{1/4} \quad \mbox{and} \quad (z^2-1)^{1/4},
\]
which are understood as analytic continuation of the corresponding functions on the real axis. 

Finally, the arccosine function appears repeatedly in \S\ref{Sasy}, including in expression \eqref{Ephiarccos} for $\varphi(z)$, definition \eqref{EpsiContz} of $\psi(z)$, definition \eqref{Elambda} of $\lambda_{\pm}(z)$, and in the expansions for the polynomials. The standard arccosine function has a branch cut on $(-\infty,-1]\cup[1,\infty)$. On the cut, the boundary values are the following, see \cite[4.23.24 \& 4.23.25]{DLMF}:
 \begin{equation}
 \arccos(z)_{\pm} = \begin{cases} \mp i\log((z^2-1)^{1/2}+z), \quad z\in[1,\infty), \\
          \pi\mp i\log((z^2-1)^{1/2}-z), \quad z\in(-\infty,-1]. \end{cases} \nonumber
 \end{equation}

\subsection{Computation of contour integrals} \label{Scontour}

Several expressions in \S\ref{Sasy} involve contour integrals around the interval $[-1,1]$, see \eqref{DinfCont}, \eqref{EpsiContz} or \eqref{cndn}. The general form of these integrals is
\begin{equation*}
	\frac{1}{2\pi i} \oint_\gamma F(\zeta)d\zeta,
\end{equation*}
where $\gamma$ encircles the interval $[-1,1]$ once in the positive direction and is contained in the region where $h(z)$ is analytic and has a positive real part, see \cite[\S 1.1]{KMcLVAV}. 

If $\log h(\zeta)$ appearing in $F(\zeta)$ is an entire function, or meromorphic with known poles, these integrals can be computed explicitly using residue calculus. For instance, if $\log h(\zeta)$ is entire, we only need to pick up the residue at infinity:
\begin{equation}
 	\frac{1}{2\pi i} \int_\gamma F(\zeta)d\zeta = -F_{-1}, \nonumber
\end{equation}
where $F_{-1}$ comes from the Taylor--Laurent series expansion
\begin{equation*}
 	\frac{-1}{t^2}F\left(\frac{1}{t}\right) \sim \sum_{m=-\infty}^\infty F_m t^m, \qquad t \rightarrow 0.
\end{equation*} 

An interesting example for which the coefficients $c_n$ and $d_n$ can be computed in this way occurs for the weight function $h(z) = \exp(-c z^{2m})$, with $m\geq 1$, see \S \ref{Sexpweight}. 

If this approach is not possible due to lack of analyticity of $F(\zeta)$, these expressions can be evaluated with the trapezoidal rule along a suitably chosen contour. This technique is exponentially accurate in the number of function evaluations, since the integrands are analytic and periodic functions, see \cite{TW}. We propose to integrate along Bernstein ellipses:
\begin{equation}\nonumber 
 E_\rho = \left\{ \tfrac12 \rho e^{i\theta} + \tfrac12 \rho^{-1} e^{-i\theta} \, | \, \theta \in [0,2\pi]\right\}, \qquad \rho \geq 1
\end{equation}
These are parameterized by a value $\rho \geq 1$, with $\rho=1$ corresponding to the interval $[-1,1]$ itself and $\rho > 1$ to an ellipse with foci at $\pm 1$. The size of the parameter $\rho$ is limited by the analyticity of the integrand in the complex neighbourhood of $[-1,1]$. It may be possible to determine an optimal value of $\rho$: we refer to \cite{Bornemann} for an extensive analysis of the optimal radius in circular Cauchy integrals and to \cite{wang} for a related study of the optimal value of $\rho$ of Bernstein ellipses in the computation of Chebyshev coefficients. For \eqref{EpsiContz}, the countour also has to encircle the point $x$ at which we wish to evaluate $\psi(x)$.

An explicit expression for the trapezoidal rule using $M$ points is
\begin{align*}
 \oint_{E_\rho} F(\zeta) d\zeta &= \int_0^{2\pi} F\left( \tfrac12 \rho e^{i\theta} + \tfrac12 \rho^{-1} e^{-i\theta} \right) \left( \tfrac12 i \rho e^{i\theta} - \tfrac12 i \rho^{-1} e^{-i\theta} \right) d\theta  \\
 &\approx \frac{2\pi k}{M} \sum_{k=0}^{M-1} F\left( \tfrac12 \rho e^{i\theta_k} + \tfrac12 \rho^{-1} e^{-i\theta_k} \right) \left( \tfrac12 i \rho e^{i\theta_k} - \tfrac12 i \rho^{-1} e^{-i\theta_k} \right)
\end{align*}
with equispaced points located at $\theta_k = 2\pi k/M$. The minimal number of points $M$ to use is of course dependent on the integrand. Due to the exponential convergence of trapezoidal rules for periodic integrands, the number $M$ can in general be taken to be fairly small, except in the vicinity of poles of the integrand. The successive doubling algorithm in \cite{Bornemann} that gives an optimal $M$ (which should increase with $n$ in $c_n$ and $d_n$) is included in the implementation.

We note that care has to be taken in general to remain on the same branch of the analytic continuation of $\log h(z)$. In other words, if we have $\textrm{Im}(\log h(\zeta))\notin (-\pi,\pi]$, evaluating $h(z)$ first and then taking the principal branch of the logarithm would not yield the correct answer. In that case, one could take $\rho$ closer to 1 and a higher $M$, or better, fill in the analytical continuation of $\log h(z)$ into the trapezoidal rules.

\subsection{On the analyticity and positivity of $h(z)$}\label{ss:h}

Several contour integrals in \S\ref{Sasy} are given in terms of the logarithm of $h$, so it is instructive to understand its possible behaviour in the complex plane. Recall that the principal branch of the logarithmic function has a branch cut along the negative real axis. 

The function $h$ satisfies several conditions stated in \cite{KMcLVAV}:
\begin{enumerate}
 \item $h$ is a real-valued and positive function on $[-1,1]$,
 \item $h$ is analytic in a complex neighbourhood of $[-1,1]$,
 \item furthermore, the real part of $h$ is strictly positive in a complex neighbourhood $U$ of $[-1,1]$. The contours in \S\ref{Sasy} are restricted to lie in $U$.
\end{enumerate}
The first condition guarantees existence of the orthogonal polynomials for all $n$. The second condition is required for the complex deformations in the Riemann-Hilbert problem to be valid.

We elaborate on the third condition. First, if $h$ vanishes at a point on $[-1,1]$, then the asymptotic behaviour of the orthogonal polynomials becomes substantially different. For examples of such behaviour, see e.g. \cite{FMFS_asymp,FMFS_Magnus}. Second, if the real part of $h$ has positive and negative values in a region, then there may be a branch cut of the principal branch of the logarithm in that region. In particular, branch points arise at roots of $h(z)$ in the complex plane. Though branch cuts may be moved, and the contours appearing in this paper may be deformed in order to avoid branch points and other singularities of $\log h$, the simplest implementation uses Bernstein ellipses confined to the region where $h(z)$ has positive real part.

It may appear to be problematic that $h(z)$ appears in the asymptotic expansions of the polynomials in the complex plane through $w(z)^{-1/2}$ when $h(z)$ has singularities there. Clearly, the polynomials do not have such singularities. However, one may verify that singularities of $h(z)$ cancel and the asymptotic expansions are, in fact, analytic functions away from the interval: see \S \ref{Sout} for example.

\subsection{Sizes of the region} \label{SsizeReg}

To conclude, we return here to Remark \ref{Rheuristics} about the sizes of the different regions of the complex plane. Since the sizes of the disks around the endpoints and the size of the lens can be chosen arbitrarily, different expansions can be valid at any given point in the complex plane. We observed experimentally from our heuristics test in the implementation that whenever different expansions are valid at a point, the corresponding relative errors in the approximation of the polynomial differ typically at most by a factor of about $2$, for large $n$.

There are a few exceptions. First, the expansion in the outer region is less accurate as we approach the interval $[-1,1]$ since it is `missing' one of the exponentials that combines into the cosine-like expression in \eqref{EpiInt}. This `missing exponential' is exponentially small in the outer region, hence it can be discarded there, but not inside the lens. The point $x=0.2+0.5i$ is not {\it on} the interval, but close to it, and indeed in Figure \ref{FNent} we only see the expected (order of) accuracy for the expansion in the outside region starting from $n=32$ for the highest number of terms in Figure \ref{FNent}. The exponentially small difference between \eqref{EpiInt} and \eqref{EpiOut} is only negligible from there onwards.

Another exception to the factor $2$ difference appears when evaluating inside a disk of radius about $0.2$ around the endpoints. There, the expansions in the respective disks can be orders of magnitude more accurate than the ones in the other regions. Indeed, the latter expansions blow up at the endpoints, whereas we note that in \eqref{Epiboun} and \eqref{Epilboun} the singularities at $z=\pm 1$ are only apparent, something that is reflected in the series implementation. We also remark that when evaluating very close to an endpoint, say at a distance $\epsilon_m^{1/3}$, where $\epsilon_m$ is the machine epsilon, one needs the series expansion \eqref{ERrlseries} of $R^{\R/\L}(z)$ (which also avoids the explicit subtraction of poles that happens in \eqref{E:simplifiedjump}), as well as a series expansion of the other factors in \eqref{Epiboun} and \eqref{Epilboun}. Without the use of series expansions, it is certainly helpful to employ the cosine transform, as commented on before in \S\ref{Sremarks}.

\section{Examples and numerical results} \label{Snum}

An important source of examples is given by the canonical modifications or perturbations of the Jacobi weight function, via polynomial or rational factors (Christoffel and Geronimus --with mass equal to $0$-- perturbations). See \cite[\S 2.7]{Ismail} for further references. In this section, we illustrate the accuracy of the asymptotic expansions with three different examples, that were chosen from literature.

\subsection{An exponential weight function} \label{Sexpweight}

Consider the weight function
\begin{equation}\nonumber 
	w(x) = h(x) = \exp\left(-cx^{2m}\right),
\end{equation}
with $\alpha=\beta=0$. This weight function appears in methods for avoiding the Gibbs phenomenon of Fourier series \cite{gelb}.

The residue calculus from \S\ref{Scontour} yields explicit formulas for the coefficients $c_n$ and $d_n$, because $\log h(z) = -c z^{2m}$ is an entire function. We have that $c_n = d_n = 0$ for $n > 2m-1$. An explicit formula for the other coefficients is
\[
	c_{n}=  -c \sum_{j = 0}^{\lfloor (2m-n-1)/2\rfloor} {j-1/2 \choose j} {2m-1-2j \choose 2m-n-1-2j}.
\]
By symmetry, we have that $d_n = (-1)^{n+1} c_n$.

Also,
\begin{equation*}
\psi(x) = \frac{1}{2}\Big(\alpha(\arccos x - \pi) +\beta \arccos x \Big) +\frac{\sqrt{1-x^2}}{2}\sum_{n=0}^{2m-1} c_n (x-1)^n,
\end{equation*}
and
\begin{equation}
	D_\infty = 2^{-\alpha/2-\beta/2}\exp\left( \frac{-c}{2} {m-1/2 \choose m} \right). \nonumber 
\end{equation}

\begin{figure}[h]
\centerline{\includegraphics[width=0.9\hsize]{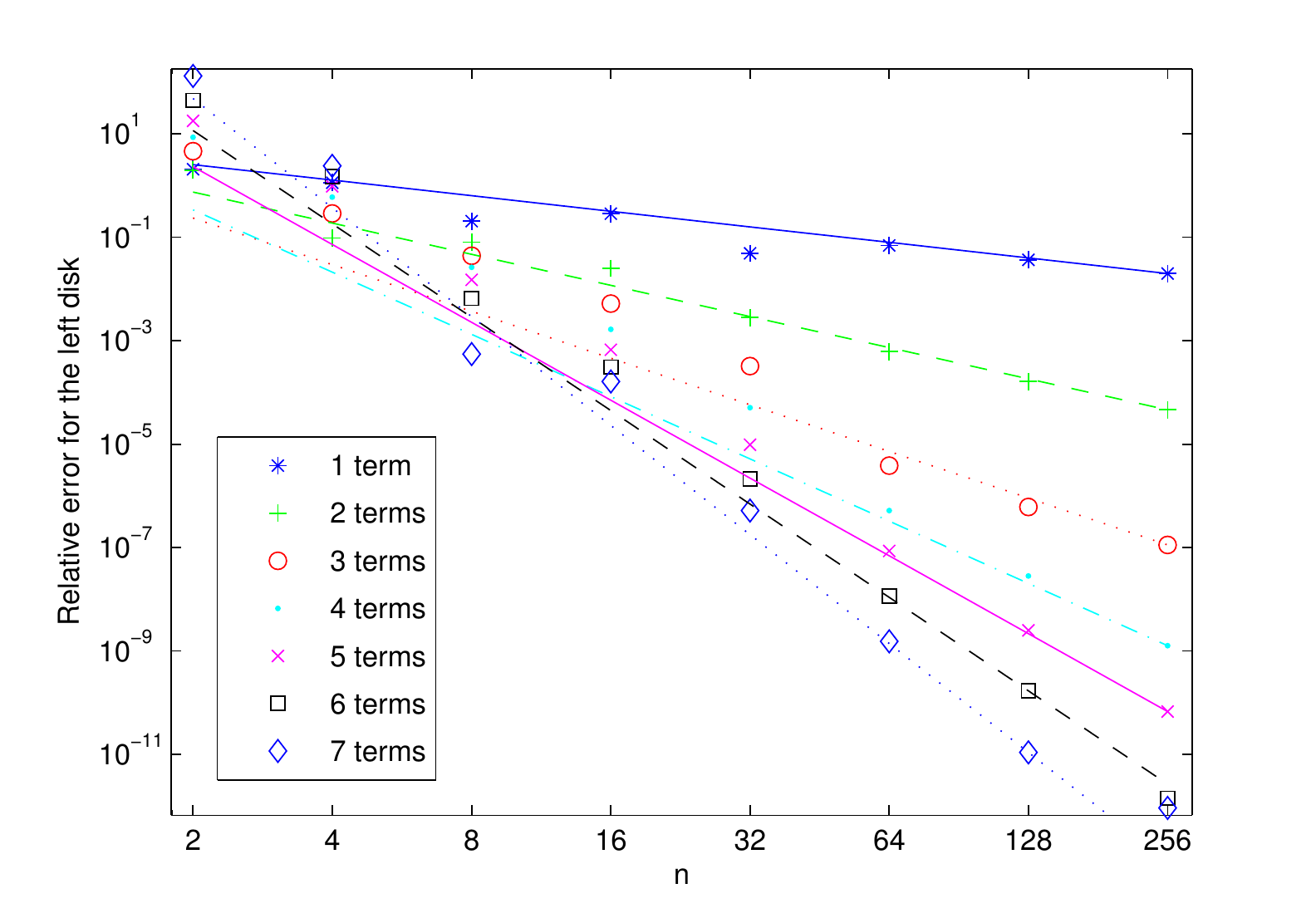} }
\caption{Relative error of the asymptotic expansion in the left boundary region as a function of $n$, for the weight function $w(x) = \exp\left(-7 x^{4}\right)$. The expansion is evaluated at  $x=-0.97$, with a varying number of terms. For each number of terms $i$, a line is plotted with slope $n^{-i}$ that interpolates the best relative error.}
\label{Fbound}
\end{figure}

We illustrate the accuracy of the asymptotic expansion in the left disk. Figure \ref{Fbound} shows the relative error at the point $x=-0.97$ as a function of $n$, for $c=7$ and $m=2$. The `exact' polynomials we compared with were obtained using Matlab routines from the OPQ-library that accompanies the book \cite{Gautschi}. It is clear from the figure that the expansions improve with increasing $n$, at a rate that depends on the number of terms. High accuracy is achieved already at moderate values of $n$, for example $10^{-7}$ relative error is seen at $n=32$ using six or seven terms. However, for small $n$, expansions with fewer terms are more accurate than expansions with more terms, as is to be expected from the asymptotic nature of the expansions. The asymptotic expansions of the coefficients $\gamma_n$, $\alpha_n$ and $\beta_n$, in the other regions and for other values of $x$ exhibit similar behaviour.

\subsection{A Jacobi-type weight function with a branch point in the complex plane} \label{Sjacweight}

Next, we consider the weight function 
\begin{equation}\nonumber
	w(x) = \frac{1}{\sqrt{(1-x)(x+3)}},
\end{equation}
which leads to $\alpha=-1/2$ and $\beta=0$. It appears in the approximation of non-periodic functions on an interval using Fourier series on a larger interval \cite[\S3]{daan}.

In this case, $\log h(z)$ is not entire due to the singularity at $x=-3$. We have used the trapezoidal rules explained in \S\ref{Scontour} in order to compute the relevant contour integrals. We chose a Bernstein ellipse with $\rho=4$, which crosses the real axis at $x=-2$: that is halfway between the singularity $x=-1$ of the integrand in $\psi(z)$ and related quantities and the singularity at $-3$. This choice reduces roundoff errors, although computing the condition number like in \cite{Bornemann} and \cite{wang} seems to give an optimal $\rho$ very close to 6 as predicted there. It would suffice to use only $M=80$ points in the trapezoidal rule, which is in between the last two iterations $M=64$ and $128$ computed here by successive doubling up to $c_2$.

\begin{figure}[h]
\centering
\subfloat{ \includegraphics[width= 0.49\textwidth]{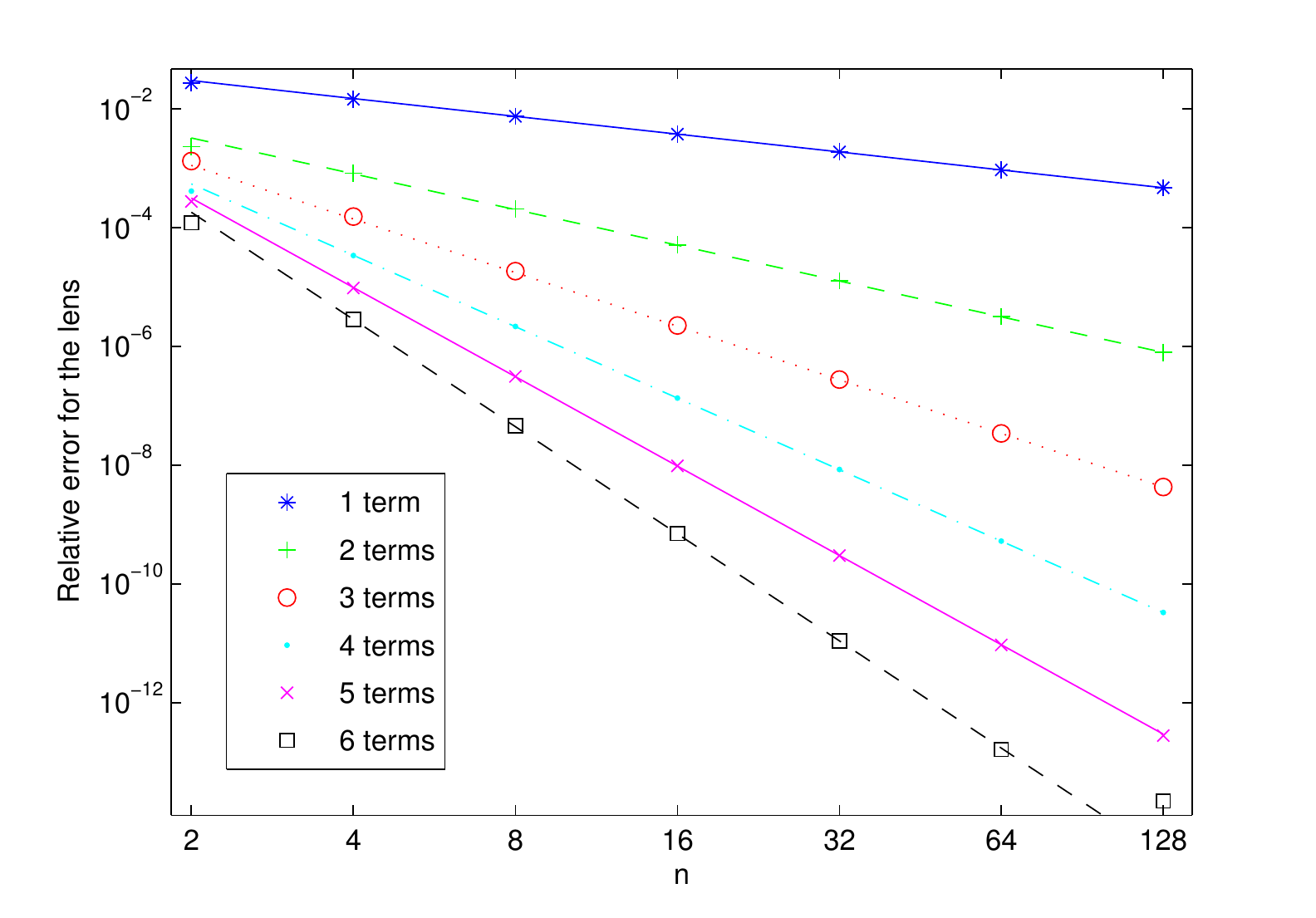}  } 
\subfloat{ \includegraphics[width= 0.49\textwidth]{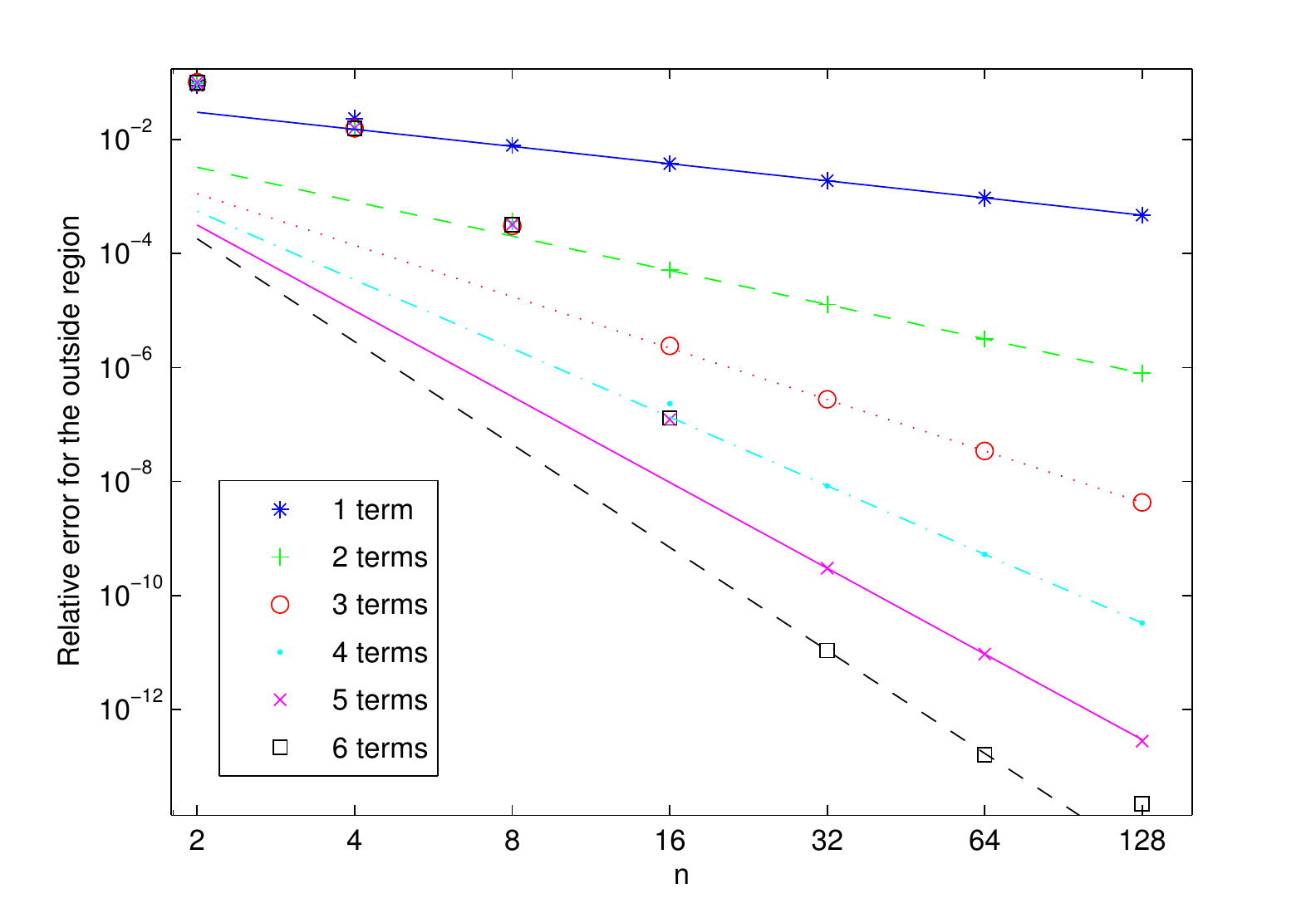}  }
\caption{Relative error of the asymptotic expansion in the lens (left) and outside region (right) as a function of the degree $n$, for the weight function $w(x) = 1/\sqrt{(1-x)(x+3)}$. Both expansions are evaluated at the same point $x=0.2+0.5i$ and for a varying number of terms.}
\label{FNent}
\end{figure}

Figure \ref{FNent} shows that we still obtain the expected order of convergence of the relative error. Some saturation appears for the highest number of terms around $10^{-13}$, due to doing computations close to machine precision and accumulating errors in the recurrence relation for the `exact' polynomials. Results are shown for the asymptotic expansion in the lens as well as in the outer region, but with both expansions evaluated at the same point $x=0.2+0.5i$. Such comparisons may lead to a decision as to which expansion to use in which part of the complex plane, see \S \ref{SsizeReg}.

\subsection{Toda measures} \label{Stoda}
Our results include the Toda modification explained in \cite[\S 2.8]{Ismail} and given by $h(x)=e^{-xt}$, with $t\in\mathbb{R}$. The resulting orthogonal polynomials appear in the literature as time--dependent Jacobi polynomials, and they have been studied in connection with integrable systems and Painlev\'e transcendents, see for instance \cite{Basor}.

For this weight function, we have that $c_0 = -t = d_0$ (which only enter in the third terms of the expansions), $c_n = 0 = d_n$ for $n \geq 1$, $D_\infty = 2^{-\alpha/2-\beta/2}$ and $\psi(x) = \frac{1}{2}\Big(\alpha(\arccos x - \pi) +\beta \arccos x \Big) -\frac{t}{2}\sqrt{1-x^2}$. The leading order term of the orthonormal polynomial in the lens is 
\begin{equation}
	p_n(x) \sim \frac{\sqrt{2}\cos\left( [n+(1+\alpha+\beta)/2]\arccos(x) -\pi/4-\alpha\pi/2 -t\sqrt{1-x^2}/2 \right)}{\sqrt{\pi}(1-x)^{\alpha/2+1/4}(1+x)^{\beta/2+1/4}\exp(-xt/2)} . \nonumber
\end{equation}

\begin{figure}[h]
\centerline{\includegraphics[width=0.7\hsize]{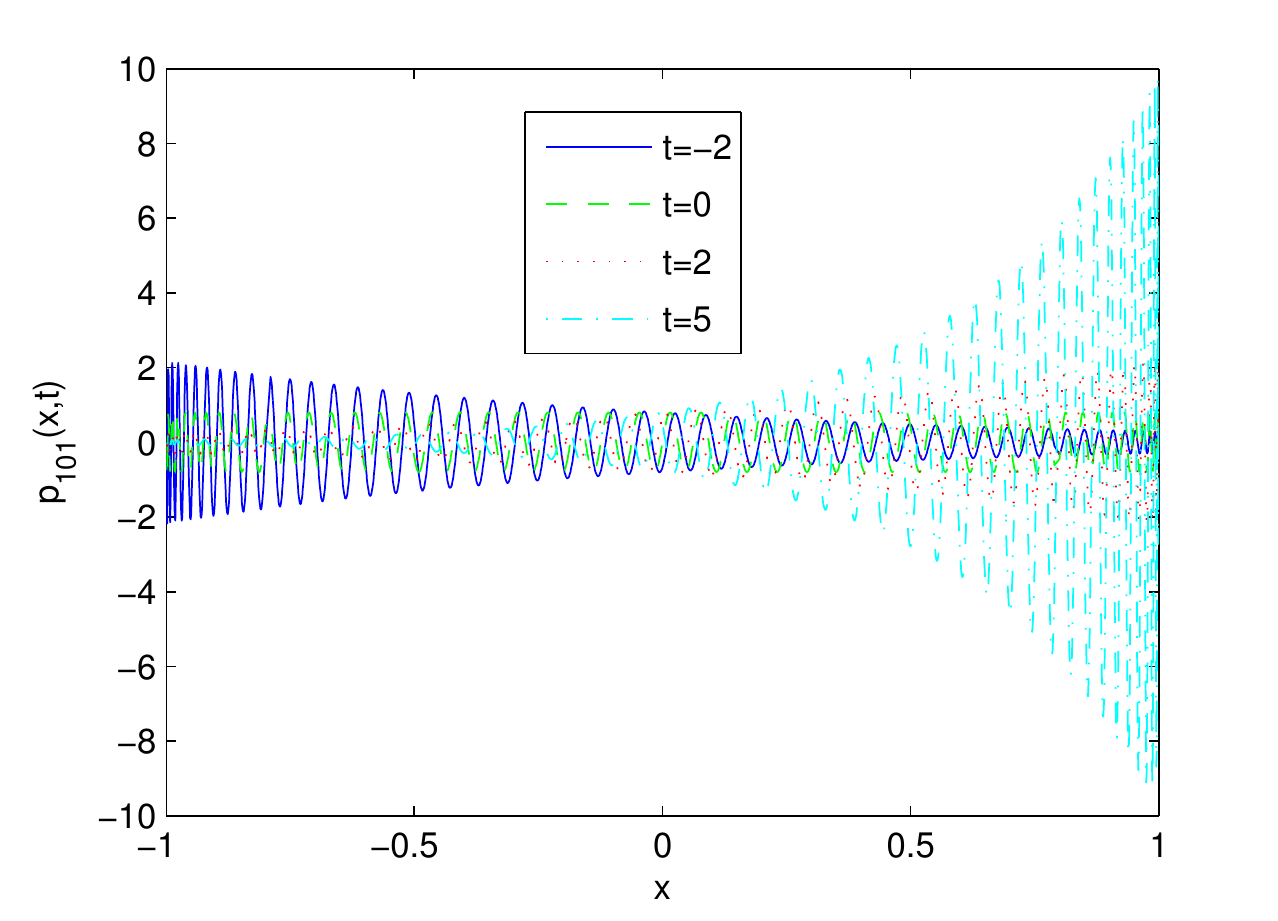} }
\caption{Leading order terms of the orthonormal polynomials in the lens for the weight function $w(x) = (1-x^2)^{-1/2}\exp(-xt)$ at $n=101$.}
\label{Ftoda}
\end{figure}

Figure \ref{Ftoda} shows us that the envelope of the polynomial indeed behaves as $e^{xt/2}$. The weight function will become very small at $x=1$ when $t \rightarrow +\infty$, making the polynomial ill-defined and large there, while $p_n(-1)$ will become very small. The inverse is true for $t \rightarrow -\infty$ and the cases $t=-2$ and $t=2$ are symmetric. In this example we have chosen $\alpha=\beta=-1/2$, and in this case we can simply use the expansion in the lens throughout the interval. The relative error with respect to the true polynomials remains bounded by $10^{-5}$ pointwise for all cases shown.




\section*{Acknowledgements}
The authors gratefully acknowledge financial support from FWO (Fonds Wetenschappelijk Onderzoek, Research Foundation - Flanders, Belgium), through FWO research project G.0617.1. The first author acknowledges support from projects MTM2012--34787 and MTM2012-36732--C03--01, from the Spanish Ministry of Economy and Competitivity. The authors thank Arno Kuijlaars, Walter van Assche and Nele Lejon for useful discussions on the topic of this paper and the anonymous reviewers for their constructive comments.

\begin{appendix}

\section{Expressions for the first four higher-order terms} \label{APP:explicit}

We illustrate the recursive computation of $R_k$ by giving the first few terms explicitly for $z$ outside the two disks. We have:
\begin{align}
	R^{\O}(z) & = I + \frac{1}{n}\left(\frac{U_{1,1}^{\R} }{z-1}+\frac{U_{1,1}^{\L}}{z+1} \right) + \frac{1}{n^2}\left(\frac{U_{2,1}^{\R} }{z-1}+\frac{U_{2,1}^{\L}}{z+1} \right)  \nonumber \\
	& + \frac{1}{n^3}\left(\frac{U_{3,1}^{\R} }{z-1} + \frac{U_{3,2}^{\R} }{(z-1)^2} +\frac{U_{3,1}^{\L}}{z+1} + \frac{U_{3,2}^{\L}}{(z+1)^2} \right) \nonumber \\
	& + \frac{1}{n^4}\left(\frac{U_{4,1}^{\R} }{z-1} + \frac{U_{4,2}^{\R} }{(z-1)^2} +\frac{U_{4,1}^{\L}}{z+1} + \frac{U_{4,2}^{\L}}{(z+1)^2} \right)+\mathcal{O}\left(\frac{1}{n^5}\right), \nonumber
\end{align}
with\footnote{Note that the following expressions are slightly different from those given in \cite[\S 8.2]{KMcLVAV}.}
\begin{align}
	U_{1,1}^{\R} = & \frac{4\alpha^2-1}{16} D_{\infty}^{\sigma_3}\begin{pmatrix}-1 & i\\ i & 1\end{pmatrix} D_{\infty}^{-\sigma_3}, \nonumber \\
	U_{1,1}^{\L} = & \frac{4\beta^2-1}{16} D_{\infty}^{\sigma_3}\begin{pmatrix} 1 & i\\ i & -1\end{pmatrix} D_{\infty}^{-\sigma_3}, \nonumber \\
	U_{2,1}^{\R} = & \frac{4\alpha^2-1}{256} D_{\infty}^{\sigma_3}\begin{pmatrix}A_2(\alpha,\beta,c_0) & iB_2(\alpha,\beta,c_0)\\ iC_2(\alpha,\beta,c_0) & D_2(\alpha,\beta,c_0)\end{pmatrix} D_{\infty}^{-\sigma_3}, \nonumber \\
	U_{2,1}^{\L} = & \frac{4\beta^2-1}{256} D_{\infty}^{\sigma_3}\begin{pmatrix} -A_2(\beta,\alpha,-d_0) & iB_2(\beta,\alpha,-d_0)\\ iC_2(\beta,\alpha,-d_0) & -D_2(\beta,\alpha,-d_0)\end{pmatrix} D_{\infty}^{-\sigma_3}, \nonumber
\end{align}
and
\begin{align}
   A_2(a,b,c) = & +8a +8b +8c -4b^2 +1, \nonumber \\ 
   B_2(a,b,c) = & -8a -8b -8c +4a^2 +4b^2 -10, \nonumber \\ 
   C_2(a,b,c) = & -8a -8b -8c -4a^2 -4b^2 +10, \nonumber \\ 
   D_2(a,b,c) = & -8a -8b -8c -4b^2 +1. \nonumber 
\end{align}

Next, we have
\begin{align}
	U_{3,1}^{\R} = & \frac{4\alpha^2-1}{8192} D_{\infty}^{\sigma_3}\begin{pmatrix}A_3(\alpha,\beta,c_0,-d_0) & i(q_3+r_3)(\alpha,\beta,c_0,-d_0)\\ i(q_3-r_3)(\alpha,\beta,c_0,-d_0) & D_3(\alpha,\beta,c_0,-d_0)\end{pmatrix} D_{\infty}^{-\sigma_3}, \nonumber \\
	U_{3,1}^{\L} = & \frac{4\beta^2-1}{8192} D_{\infty}^{\sigma_3}\begin{pmatrix} -A_3(\beta,\alpha,-d_0,c_0) & i(q_3+r_3)(\beta,\alpha,-d_0,c_0)\\ i(q_3-r_3)(\beta,\alpha,-d_0,c_0) & -D_3(\beta,\alpha,-d_0,c_0)\end{pmatrix} D_{\infty}^{-\sigma_3}, \nonumber
\end{align}
with
\begin{align}
   A_3(a,b,f,g) = & 16(4b^2-1)(f+g+2a+2b)-2(4b^2-1)(2a^2+2b^2-1)\nonumber \\
	& -128[(a+b)^2+f(f+2a+2b)],\nonumber \\
   q_3(a,b,f,g) = & 128(a+b)^2+128f(f+2a+2b)-\tfrac{388}{3}a^2-84b^2+\tfrac{64}{3}a^4+16b^4 \nonumber \\
	& +48a^2b^2+176,\nonumber \\
   r_3(a,b,f,g) = & -128(a+b)(a^2+b^2)+320(a+b)-64b^2(f+g)-128fa^2\nonumber\\
		  &+304f+16g,\nonumber \\
   D_3(a,b,f,g) = & 16(4b^2-1)(f+g+2a+2b) + 2(4b^2-1)(2a^2+2b^2-1)\nonumber \\
	&+128[(a+b)^2+f(f+2a+2b)],\nonumber
\end{align}
and
\begin{align}
	U_{3,2}^{\R} = & \frac{(4\alpha^2-1)(4\alpha^2-9)(4\alpha^2-25)}{12288}D_\infty^{\sigma_3} \begin{pmatrix} -1 & i \\ i & 1 \end{pmatrix} D_\infty^{-\sigma_3}, \nonumber \\
	U_{3,2}^{\L} = & \frac{(4\beta^2-1)(4\beta^2-9)(4\beta^2-25) }{12288}D_\infty^{\sigma_3} \begin{pmatrix} -1 & -i \\ -i & 1 \end{pmatrix} D_\infty^{-\sigma_3}, \nonumber \\
	U_{4,1}^{\R} = & \frac{4\alpha^2-1}{65536}D_\infty^{\sigma_3}
	\begin{pmatrix}(v_4+w_4)(\alpha,\beta,c_0,d_0,c_1) & i(x_4+y_4)(\alpha,\beta,c_0,d_0,c_1)\\ 
	i (x_4-y_4)(\alpha,\beta,c_0,d_0,c_1) & (v_4-w_4)(\alpha,\beta,c_0,d_0,c_1)\end{pmatrix}D_\infty^{-\sigma_3} , \nonumber \\
	U_{4,1}^{\L} = & \frac{4\beta^2-1}{65536} D_\infty^{\sigma_3} \begin{pmatrix} -(v_4+w_4)(\beta,\alpha,-d_0,-c_0,d_1) & i(x_4+y_4)(\beta,\alpha,-d_0,-c_0,d_1)\\ 	i(x_4-y_4)(\beta,\alpha,-d_0,-c_0,d_1) & -(v_4-w_4)(\beta,\alpha,-d_0,-c_0,d_1)\end{pmatrix}D_\infty^{-\sigma_3}, \nonumber
\end{align}
with
\begin{align*}
v_4(a,b,f_0,g_0,f_1) = & \frac{1-4b^2}{6}\left[384(f_0^2+g_0^2-f_0g_0+3(a+b)(f_0-g_0)	)\right.\\
	&\left.+16((a^2+b^2)^2+a^2b^2)+1196(a+b)^2-88ab-219\right],\\[2mm]
w_4(a,b,f_0,g_0,f_1) = & -4\left[(4b^2-1)(8b^2+4a^2-11)g_0+48(4a^2-9)f_1-768abf_0\right.\\
	&\left.-f_0(128f_0(f_0+3(a+b))+16b^2(b^2+2a^2+25)+312a^2+139)\right.\\
	&\left.-2(a+b)(b^2(24b^2+24a^2+46)+58a^2+128ab+3)\right],\\[2mm]
x_4(a,b,f_0,g_0,f_1) = & 4\left[(4b^2-1)(8b^2+12a^2-29)g_0+48(4a^2-9)f_1-768abf_0\right.\\
	&\left.-f_0(128f_0(f_0+3(a+b))+16b^2(b^2+6a^2+16) + 8a^2(8a^2-7)+643)\right.\\
	&\left.-4(a+b)(b^2(12b^2+36a^2-31)+a^2(16a^2-65)+64ab+132)\right],\\[2mm]
y_4(a,b,f_0,g_0,f_1) = & \frac{4}{3}\left[48(4b^2-1)g_0(g_0-f_0-3(a+b))+b^4(48a^2+542)\right.\\
	&\left.+48f_0^2(4b^2+12a^2-28)+144(a+b)(4b^2+8a^2-19)f_0+a^4(56b^2+422)\right.\\
	&\left.+ab(1152(a^2+b^2)+988ab-2880)+16a^6-1393b^2-951a^2-498+8b^6\right].
\end{align*}

Finally,
\begin{align}
  U_{4,2}^{\R} = & \frac{(4\alpha^2-1)(4\alpha^2-9)(4\alpha^2-25)}{2^{17}3}D_\infty^{\sigma_3}
	\begin{pmatrix} E_4 & iF_4\\ iG_4 & H_4 \end{pmatrix}D_\infty^{-\sigma_3}, \nonumber \\
   E_4 = & -4\alpha^2-8\beta^2+48c_0+48\alpha+48\beta+3, \nonumber \\
   F_4 = & +8\alpha^2+8\beta^2-48c_0-48\alpha-48\beta-52, \nonumber \\
   G_4 = & -8\alpha^2-8\beta^2-48c_0-48\alpha-48\beta+52, \nonumber \\
   H_4 = & -4\alpha^2-8\beta^2-48c_0-48\alpha-48\beta+3, \nonumber \\
	U_{4,2}^{\L} = & \frac{(4\beta^2-1)(4\beta^2-9)(4\beta^2-25) }{2^{17}3}
  D_\infty^{\sigma_3} 
  \begin{pmatrix} I_4 & iJ_4 \\ iK_4 & L_4 \end{pmatrix}D_\infty^{-\sigma_3}, \nonumber \\
   I_4 = & -8\alpha^2-4\beta^2-48d_0+48\alpha+48\beta+3, \nonumber \\
   J_4 = & -8\alpha^2-8\beta^2-48d_0+48\alpha+48\beta+52, \nonumber \\
   K_4 = & +8\alpha^2+8\beta^2-48d_0+48\alpha+48\beta-52, \nonumber \\
   L_4 = & -8\alpha^2-4\beta^2+48d_0-48\alpha-48\beta+3. \nonumber
\end{align}



\section{Proof of Proposition \ref{Tsimpl}}\label{APP:proof_of_simplifications}

The basic idea for the proof is induction, and also the fact that the term 
\begin{equation}
\left(s_j(z)\Delta_{m-j}(z) + s_{m-j}(z)\Delta_j(z)\right)^{\R/\L} \nonumber
\end{equation}
always simplifies to 0 when $m$ is odd or something proportional to $I$ when $m$ is even.

\begin{proof}
For simplicity we present the case of $s_m^{\R}(z)$ and omit the superscripts. 
The proof proceeds by induction in $m$ in formula \eqref{EcheckS}. The case $m=1$ is clear from
\eqref{EcheckS} because then there is an empty summation in \eqref{EcheckS}. We assume that the proposition holds for $k=1,\ldots, m-1$.

If $m$ is even, we can rewrite the sum on the right hand side of \eqref{EcheckS} as follows:
\begin{equation}
 \sum_{j=1}^{m-1} s_j(z)\Delta_{m-j}(z) = s_{m/2}(z)\Delta_{m/2}(z) + \sum_{j=1}^{m/2-1} \left[s_j(z)\Delta_{m-j}(z)+s_{m-j}(z)\Delta_{j}(z)\right], \nonumber
\end{equation}
and if $m$ is odd,
\begin{equation}
 \sum_{j=1}^{m-1} s_j(z)\Delta_{m-j}(z)=\sum_{j=1}^{(m-1)/2} \left[s_j(z)\Delta_{m-j}(z)+s_{m-j}(z)\Delta_{j}(z)\right]. \nonumber
\end{equation}

We recall the formula \eqref{EDeltak}, and for $1\leq k\leq m-1$, we denote
\begin{equation}
 A_k^{\alpha}=\begin{pmatrix} \tfrac{(-1)^k}{k}(\alpha^2+\tfrac{k}{2} -\tfrac{1}{4}) & -i\left(k-\tfrac{1}{2}\right) \\[1mm] (-1)^k\left(k-\tfrac{1}{2}\right)i & \tfrac{1}{k}(\alpha^2 +\tfrac{k}{2} -\tfrac{1}{4}) \end{pmatrix}, \nonumber
\end{equation}
and
\begin{equation}
 B_k^{\alpha}=
\begin{cases}
 A_k^{\alpha},& \qquad k\,\, \textrm{odd},\\
A_k^{\alpha}-\displaystyle\frac{4\alpha^2+2k-1}{2k} I,& \qquad k\,\,\textrm{even}.
\end{cases} \nonumber
\end{equation}
Then
\begin{align}
 s_j(z) \Delta_{m-j}(z) + & s_{m-j}(z) \Delta_j(z)  =  \frac{(\alpha,j-1)(\alpha,m-j-1)}{\left(2\log \varphi(z)\right)^{m}} \nonumber \\
	& \times D_{\infty}^{\sigma_3} M(z) F(z)^{\sigma_3} C^{\alpha}_{j,m-j} F(z)^{-\sigma_3}M(z)^{-1}D_{\infty}^{-\sigma_3}, \nonumber
\end{align}
where $C^{\alpha}_{j,m-j}=B_j^{\alpha} A_{m-j}^{\alpha}+B_{m-j}^{\alpha} A_j^{\alpha}$.

We observe that if $m$ is odd then for $j$ odd we have $m-j$ even, and for $j$ even we have $m-j$ odd. In both cases, a direct computation shows that $C^{\alpha}_{j,m-j}=0$. Then $s_m(z)=\Delta_m(z)$, and the proposition is true in this case. 

If $m$ is even, then $j$ and $m-j$ are simultaneously odd or even. In this case, the matrix
$C^{\alpha}_{j,m-j}$ reduces to a multiple of the identity matrix:
\begin{equation}
 C^{\alpha}_{j,m-j}=\lambda_{j,m-j} I, \nonumber
\end{equation}
where
\begin{equation}
\lambda_{j,m-j}=(-1)^{j-1}\frac{(4\alpha^2+4j(m-j)-1)(4\alpha^2-(2j-1)(2(m-j)-1))}{8j(m-j)}.\nonumber
\end{equation}
Therefore, 
\begin{equation}\label{sum0}
 s_j(z) \Delta_{m-j}(z)+s_{m-j}(z) \Delta_j(z)=
 \frac{(\alpha,j-1)(\alpha,m-j-1)}{\left(2\log \varphi(z)\right)^{m} }\lambda_{j,m-j} I.
\end{equation}

Now, we need to sum this last expression over $j$. Using the symmetry $j\leftrightarrow m-j$ of the coefficients, we can write
\begin{align}
&s_{m/2}(z)\Delta_{m/2}(z) + \sum_{j=1}^{m/2-1} \left[s_j(z)\Delta_{m-j}(z)+s_{m-j}(z)\Delta_{j}(z)\right]  \nonumber  \\
&=\frac{1}{2}\sum_{j=1}^{m-1} \left[s_j(z)\Delta_{m-j}(z)+s_{m-j}(z)\Delta_{j}(z)\right]  \nonumber  \\
&=\frac{1}{2\left(2\log \varphi(z)\right)^{m}}\sum_{j=1}^{m-1} (\alpha,j-1)(\alpha,m-j-1) \lambda_{j,m-j}I, \label{EhypSum}
\end{align}
using \eqref{sum0}. The sum \eqref{EhypSum} is of hypergeometric type, since we can write the coefficients $(\alpha,k)$ as follows:
\begin{equation}\label{coefhypergeom}
(\alpha,k)=(-1)^k \frac{\left(\tfrac{1}{2}+\alpha\right)_k \left(\tfrac{1}{2}-\alpha\right)_k}{k!},
\end{equation}
in terms of the standard Pochhammer symbol. Next, we may apply a known algorithm due to Gosper and later extended by Zeilberger, see \cite{gosper} and also \cite{Koornwinder}. Let 
\begin{equation}
a_j=(\alpha,j-1)(\alpha,m-j-1)\lambda_{j,m-j}, \nonumber
\end{equation}
then the algorithm seeks $S_j$, such that the sum telescopes as follows:
\begin{equation}
\sum_{j=1}^{m-1} a_j=S_{m-1}-S_0. \nonumber
\end{equation}
Under the hypothesis that $S_j/S_{j-1}$ is a rational function in $j$, the ratio $a_j/a_{j-1}$ is also rational in $j$,
and can be written as
\begin{equation}
\frac{a_j}{a_{j-1}}=\frac{p_j}{p_{j-1}}\frac{q_j}{r_j}, \nonumber
\end{equation}
where in this case
\begin{equation}
\begin{aligned}
p_j&=(4\alpha^2-4j^2+4jm-1)(4\alpha^2+4j^2-4jm+2m-1), \nonumber \\
q_j&=(-m+j-1)(2\alpha+2j-3)(2\alpha-2j+3), \nonumber \\
r_j&=(2\alpha+2(m-j)-1)(2\alpha-2(m-j)+1)j. \nonumber
\end{aligned}
\end{equation}

Then $S_j$ is constructed as follows:
\begin{equation}\label{Sjaj}
S_j=\frac{q_{j+1}}{p_j}f_j a_j,
\end{equation}
where $f_j$ is a function of $j$ to be determined. In this case, $f_j$ is a polynomial of 
degree $2$ in $j$ that satisfies the linear recursion
\begin{equation}
p_j=q_{j+1}f_j-r_j f_{j-1}. \nonumber
\end{equation}
From here we get
\begin{equation}
f_j=-\frac{4\alpha^2-4j^2+4jm-4j+2m-1}{m}, \nonumber
\end{equation}
and using \eqref{Sjaj}, we get $S_j$. Finally, a brief computation yields
\begin{equation} \nonumber
\begin{aligned}
S_{m-1}-S_0 & = -\frac{4\alpha^2+2m-1}{m!}\left(\tfrac{1}{2}+\alpha\right)_{m-1} \left(\tfrac{1}{2}-\alpha\right)_{m-1}\\
	&= \frac{4\alpha^2+2m-1}{m}(\alpha,m-1),
\end{aligned}
\end{equation}
using \eqref{coefhypergeom} and the fact that $m$ is even. This completes the proof for $s_k^{\R}$.

This reasoning can be carried out analogously for the $s_k^{\L}$, by replacing $\alpha$ by $\beta$, $\ln(\varphi(z))$ by $\ln(-\varphi(z))$ and taking into account the extra sign of the off-diagonal elements in $A_k^{\alpha}$. 
\end{proof}

\end{appendix}

\bibliographystyle{abbrv}
\bibliography{strong_asymptoticsv41.bbl}

\end{document}